\documentclass[10pt,twocolumn]{IEEEtran}

\usepackage{amsmath, graphics,amssymb,epsfig,subfigure,color}
\usepackage{algorithm}
\usepackage{algorithmic}
\usepackage{float}
\usepackage{multirow}
\usepackage{cuted,flushend}
\usepackage{midfloat}
\usepackage{cite}

\newtheorem{theorem}{Theorem}

\newtheorem{corollary}{Corollary}

\def\b{\mathbb}

\def\c{\mathcal}
\def\sinr{\textrm{SINR}}

\def\b{\mathbb}
\def\c{\mathcal}
\def\d{{\rm d}}

\newtheorem{lem}{Lemma}

\begin{document}
\title{Power Control for D2D Underlaid Cellular Networks: Modeling, Algorithms and Analysis}
\author{Namyoon Lee, Xingqin Lin, Jeffrey G. Andrews, and Robert W. Heath Jr.  \bigskip
\\
\normalsize Wireless Networking and Communication Group \\ \normalsize Department of Electrical and Computer Engineering
\\ \normalsize The University of Texas at
Austin, Austin, TX 78712 USA\\
      { \normalsize E-mail~:~\{namyoon.lee, xlin, rheath\}@utexas.edu}, jandrews@ece.utexas.edu}
%\author{\normalsize  Namyoon Lee$^\dagger$, Robert Heath Jr.\bigskip
%\\
%\normalsize Wireless Networking and communication Group\\
%\normalsize Department of Electrical and Computer Engineering
%\\ \normalsize The University of Texas at
%Austin, Austin, TX 78712 USA\\
%      { \it \normalsize E-mail~:~namyoon.lee@utexas.edu and rheath@ece.utexas.edu} }

%\date{}
\maketitle
\begin{abstract}
%
%
%
%Device-to-device (D2D) communication is a promising concept to increase the capacity of future cellular networks. In spite of its promising gains, the characterization of the performance in terms of relevant system parameters has remained unknown to date. This paper makes progress in this direction by characterizing the performance of power control methods.

This paper proposes a random network model for a D2D underlaid cellular system using \textit{stochastic geometry} and develops centralized and distributed power control algorithms. The goal of the centralized power control is two-fold:  ensure the cellular users have sufficient coverage probability by limiting the interference created by underlaid D2D users, while scheduling as many D2D links as possible. For the distributed power control method, the optimal on-off power control strategy is proposed, which maximizes the sum rate of D2D links. Analytical expressions are derived for the coverage probabilities of cellular, D2D links, and the sum rate of the D2D links in terms of the density of D2D links and the path-loss exponent. The analysis reveals the impact of key system parameters on the network performance. For example, the bottleneck of D2D underlaid cellular networks is the cross-tier interference between D2D links and the cellular user, not the D2D intra-tier interference when the density of D2D links is sparse. Simulation results verify the exactness of the derived coverage probabilities and the sum rate of D2D links. \end{abstract}

%\begin{IEEEkeywords}
%Power control, device-to-device communication, cellular networks, Poisson point process, stochastic geometry.
%\end{IEEEkeywords}
%
%\begin{keywords}
%MIMO multi-user interference relay channels, Degrees of Freedom
% Interference Neutralization
%\end{keywords}

%\IEEEpeerreviewmaketitle
%\newtheorem{thm}{Theorem}%[section]
%\newtheorem{cor}[thm]{Corollary}
%\newtheorem{lem}[thm]{Lemma}
%\newtheorem{rem}{Remark}
%%\newtheorem{defi}[thm]{Definition}
%\newtheorem{defi}{Definition}
%\newtheorem{prop}{Proposition}
%\newpage

%%%%%%%%%%%%%%%%%%%%%%%%%%%%%%%%%%%%%%%%%%%%%%%%%%%%%%%%%%%%%%%%%%%%%%%%%%%%%%%%%%%%%%%%%%%%%%%%%%
\vspace{-0.5cm}
\section{Introduction}
%\nocite{*}

Device-to-device (D2D) communication underlaid with cellular networks allows direct communication between mobile users \cite{Doppler1:09, Fodor:12, 3gppD2D}. D2D is an attractive approach for dealing with local traffic in cellular networks. The initial motivation for incorporating D2D communication in cellular networks is to support proximity-based services, e.g. social networking applications or media sharing \cite{3gppD2D}. Assuming there are proximate communication opportunities, D2D communication may also increase  area spectral efficiency, improve cellular coverage, lower end-to-end latency, or reduce handset power consumption \cite{Richardson:10}, \cite{Fodor:12}. In spite of these potential gains, the coexistence of D2D and cellular communication in the same spectrum is challenging due to the difficulty of interference management \cite{Fodor:12}. Specifically, the underlaid D2D signal becomes a new source of interference. As a result, cellular links experience cross-tier interference from the D2D transmissions whereas the D2D links need to combat not only the inter-D2D interference but also the cross-tier interference from the cellular transmissions. Therefore, interference management is essential to ensure successful coexistence of cellular and D2D links.

Power control is an effective approach to mitigate interference in wireless networks; it is broadly used in current wireless systems. In this paper, we propose power control methods for interference coordination and analyze their performance in D2D underlaid cellular networks. In particular, we consider a hybrid random network model using stochastic geometry and develop two different power control algorithms for the proposed network model. With a carefully designed (centralized) power control technique, we show that multiple D2D links may communicate successfully while guaranteeing reliable communication for the existing cellular link. This shows that, with an appropriate power control technique, underlaid  D2D links help to increase the network sum-throughput  without causing unacceptable performance degradation to existing cellular links.

\vspace{-0.2cm}
\subsection{Related Work}\vspace{-0.01cm}
There has been considerable interest in power control techniques for D2D underlaid cellular networks. A simple power control scheme was proposed in \cite{Yu2:12} for a single-cell scenario and deterministic network model, which regulates D2D transmit power to protect the existing cellular links. To maximize the sum rate of the network, a D2D transmit power allocation method was proposed in \cite{Yu3:11} for the deterministic network model. A dynamic power control mechanism for a single D2D link communication was proposed in \cite{Gu:12}, which targets improving the cellular system performance by mitigating the interference generated by D2D communication. The main idea was to adjust the D2D transmit power via base station (BS) to protect cellular users. In \cite{Xiao:11}, a power minimization solution with joint subcarrier allocation, adaptive modulation, and mode selection was proposed to guarantee the quality-of-service demand of D2D and cellular users. In prior work \cite{Fitzek:06,Wu:01,Hsieh:04,Janis1:09,Doppler1:09,Yu2:12,Gu:12,Xiao:11,Kaufman:08,Fodor:12,Lei:12, Janis2:09,Yu4:09, chandrasekhar2009power},
D2D power control strategies are developed and evaluated in a deterministic D2D link deployment scenario. For a random network model, spectrum sharing between ad hoc and cellular networks was studied in \cite{Huang:07, lee2011spectrum, lin2013optimal} but power control -- an essential component of D2D underlaid cellular networks --  has not been addressed. Power control has been studied in other random ad hoc networks without considering cellular networks (see e.g. \cite{baccelli2003downlink, jindal2008fractional, net:Zhang12tcom}). In our paper, we propose power control algorithms and analyze their performance in a D2D underlaid cellular network.

%\vspace{-0.1cm}
\subsection{Contributions}
In this paper, we consider a D2D underlaid cellular network in which an uplink cellular user intends to communicate with the BS while multiple D2D links coexist in the common spectrum. In such a network, we model the D2D user's (transmitter's) locations using a spatial Poisson point process (PPP). The rationale is that stochastic geometry is an useful tool to model irregular spatial structure of D2D locations and analytically quantify the interference in D2D underlaid cellular networks. In this D2D underlaid cellular system, we propose a centralized and a distributed power control algorithm. The main idea of the centralized algorithm is to design the transmit power of mobile users so as to maximize the signal-to-interference-plus-noise ratio (SINR) of the cellular link while satisfying the individual target SINR constraints for D2D links. Using the fact that the centralized power allocation problem is convex, we solve it with a feasibility set increment technique. A main observation is that the centralized power control approach is possible to significantly improve the overall cellular network throughput due to the newly underlaid D2D links while guaranteeing the coverage probability of pre-existing cellular links.

We also propose a simple distributed on-off power control algorithm. Note that the centralized algorithm requires global channel state information (CSI) possibly at  a centralized controller, which may incur high CSI feedback overhead. To resolve this issue,  the proposed on-off power control method requires CSI knowledge about the direct link between the transmitter and its corresponding receiver only. In particular, for the distributed power control method, we derive analytic expressions including the coverage probabilities of both cellular and D2D links and the sum rate of D2D links. One important insight obtained from the analysis is that  on-off power control strategy for the uplink user is actually  optimal in terms of the coverage probability of the cellular link, agreeing with the finding in ad hoc networks \cite{net:Zhang12tcom}. Further, we derive the optimal D2D transmission probability which maximizes the sum rate of D2D links when the distributed on-off power control algorithm is used. In contrast to the centralized power control method, the distributed power control algorithm is not sufficient to guarantee  reliable cellular communication, though it does improve the cellular network throughput by additional D2D communication. We verify the results by simulating two different D2D link deployment scenarios.

The remainder  of this paper is organized as follows. In Section II, the proposed model for D2D underlaid cellular networks is described. The proposed power control algorithms are presented in Section III. In Section IV, the analytical expressions for the coverage probabilities of the cellular and typical D2D link are derived and validated through comparison with the simulation results. For the distributed power control, the sum rate of D2D links is derived in Section V. Simulation results are provided in Section VI to compare the performance of the proposed algorithms, which are followed by our conclusions in Section VII.

\section{System Model}

In this section, we  present the system model and describe network metrics that will be used in this paper.

%\subsection{System Model}

\begin{figure}
\centering
\includegraphics[width=9cm]{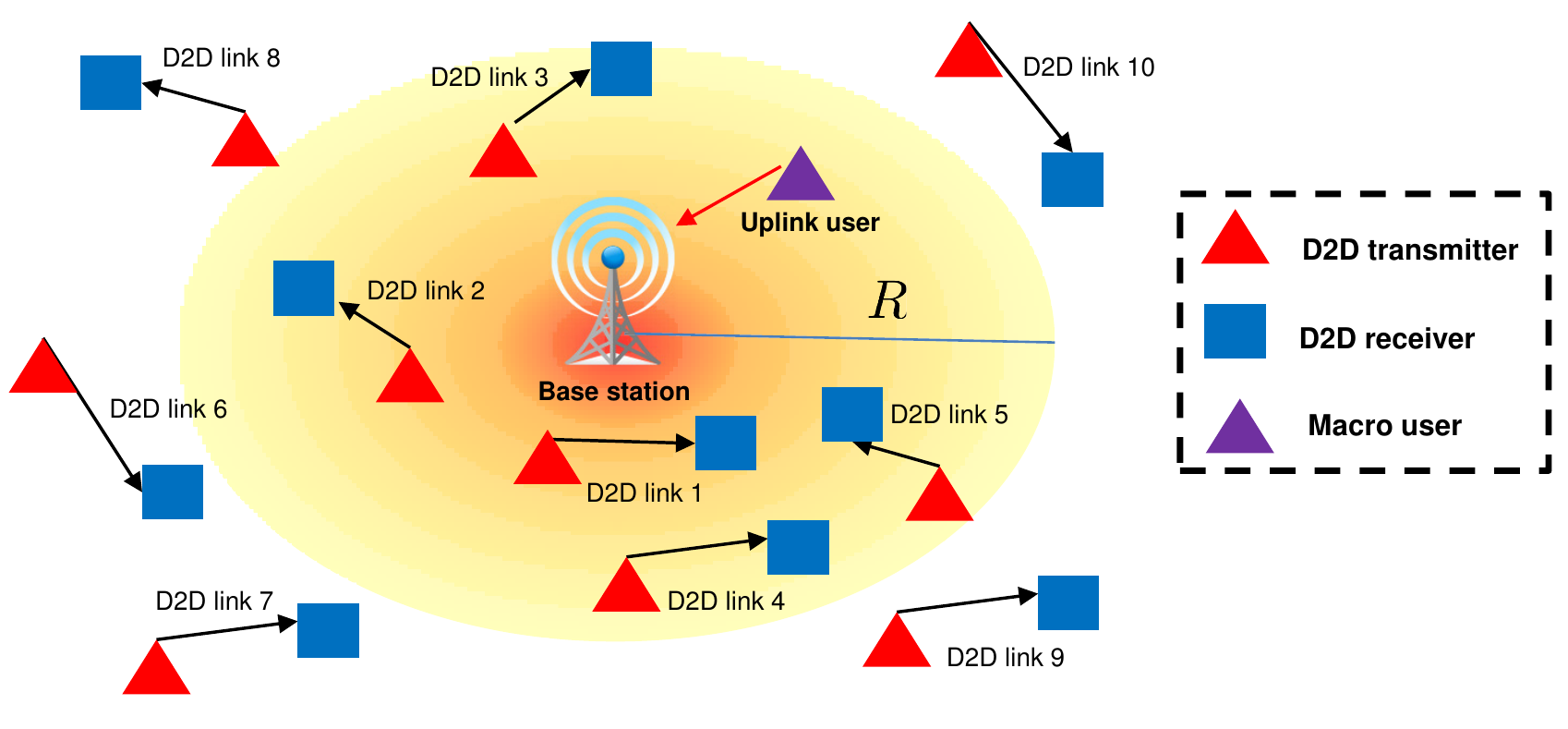}
\caption{A single-cell D2D underlaid cellular system: one macro user establishes a cellular link with the BS while five active D2D links are established in a circular disk centered at the BS and with radius $R$. In this model, the active D2D links outside the circular disk are considered as out-of-cell D2D interference; however, out-of-cell interference from the marcro users belonging to the other cells is ignored.}
\label{fig:scenario}\vspace{-0.6cm}
\end{figure}

We consider a D2D underlaid cellular network, as shown in Fig. \ref{fig:scenario}. In this model, let the circular disk $\mathcal{C}$ with radius $R$ denotes the coverage region of a BS centered at the origin. We assume that one cellular uplink user is uniformly located in this region. Further, we assume that the locations of the D2D transmitters are distributed in the whole $\mathbb{R}^2$ plane according to a homogeneous PPP $\Phi$ with density $\lambda$. The associated receiver with a D2D transmitter is located at a fixed distance away with isotropic direction. We assume all nodes have one antenna.

Under the given assumptions, the number of D2D transmitters in $\mathcal{C}$ is a Poisson random variable with mean $\b E[K]= \lambda \pi R^2$. Given a particular realization of the PPP $\Phi$, the received signals at D2D receiver $k$ and the BS are written as
\begin{align}
{y}_k =& {h}_{k,k}d_{k,k}^{-\frac{\alpha}{2}}{s}_{k}\!+\! {h}_{k,0}d_{k,0}^{-\frac{\alpha}{2}}{s}_{0} \!+\!\!\!\!\! \sum_{\ell=1, \ell\neq k}^{K}\!\!\!\!\!{h}_{k,\ell}d_{k,\ell}^{\frac{\alpha}{2}}{s}_{\ell}\!+\! n_{k} ,    \\
{ y}_{0} =&{ h}_{0,0}d_{0,0}^{-\alpha/2}s_{0}+\sum_{k=1}^{K}{ h}_{0,k}d_{0,k}^{-\alpha/2}s_{k} + { n}_{0},
\end{align}
where subscript $0$ is used for the uplink signal to the BS and subscript $k, k\neq 0,$ are used for D2D links; $s_k$ and ${s}_{0}$ denote the signal sent by D2D transmitter $k$ and the uplink user; $y_k$ and ${ y}_{0}$ represent the received signal at D2D receiver $k$ and the BS;
${n}_k$ and ${ n}_{0}$ denote the additive noise at D2D receiver $k$ and the BS distributed as $\mathcal{CN}(0,\sigma^2)$; ${{h}_{k,\ell}}$ and ${ h}_{0,k}$
represent the distance-independent fading from D2D transmitter $\ell$ to receiver $k$ and the channel from D2D transmitter $k$ to the BS, and are independently distributed as $\mathcal{CN}(0,1)$. Here, we assume the distance dependent path-loss model, i.e., $d_{k,\ell}^{-\alpha}$ for all $k,\ell$ where ${d}_{k,j}$ denotes the distance from transmitter $j$ to receiver $k$ and $\alpha$ is the path-loss exponent. The transmit power satisfies the peak power constraints, i.e., $|{s}_{0}|^2\leq P_{\rm max, c}$, $|{s}_{k}|^2\leq P_{\rm max, d}$ for $k\in \{1,2,\ldots,K\}$.

%$\mathbb{E}\left[|{s}_{0}|^2\right] \leq P_{\textrm{avg, c}}$, $\mathbb{E}\left[|{s}_{k}|^2\right] \leq P_{\textrm{avg, d}}$

Then the SINR at D2D receiver $k$ and the BS are given by
\begin{eqnarray}
\textrm{SINR}_k(K,{\bf p})&=&\frac{|h_{k,k}|^2d_{k,k}^{-\alpha}p_k}{|h_{k,0}|^2d_{k,0}^{-\alpha}p_0 +\sum_{\ell\neq k}^{K}|h_{k,\ell}|^2d_{k,\ell}^{-\alpha}p_{\ell} +\sigma^2},  \\
\textrm{SINR}_0(K,{\bf p})&=&\frac{|{ h}_{0,0}|^2d_{0,0}^{-\alpha}p_0}{\sum_{k=1}^{K}|{h}_{0,k}|^2d_{0,k}^{-\alpha}p_{k} +\sigma^2},
 \end{eqnarray}
where ${\bf p}=[p_0,p_1,\ldots,p_K]^T$ denotes transmit power profile vector with $p_i$ being the transmit power of transmitter $i$.

Note that our system model ignores out-of-cell interference from macro users in the other cells. Nevertheless, the proposed model is able to capture the effect of dominant interference, which mainly determines the network performance. For the case of the uplink user transmission, the dominant interferer is the nearest D2D transmitter in the cell because the interference power of the nearest D2D transmission is strong than the power of the out-of-cell interference with high probability. For the case of the typical D2D transmission, the dominant interferer to any D2D link is either the uplink transmission or the nearest D2D transmission in the cell. Therefore, our system model captures the dominant interference effect for the both uplink and D2D links; it is possible to offer a tight upper performance of the D2D underlaid cellular system.

%\subsection{Performance Metrics}

We are interested in the coverage probability of the cellular link and D2D links. The cellular coverage probability is defined as
\begin{eqnarray}
\bar{P}^{(C)}_{\textrm{cov}} (\beta_0) = \b E[ {P}^{(C)}_{\textrm{cov}} ( {\bf p}, \beta_0) ] =  \b E[\mathbb{P}(\textrm{SINR}_0(K,{\bf p}) \geq \beta_0)],
\end{eqnarray}
where $\beta_0$ represents the minimum SINR value for reliable uplink connection. Similarly, the D2D coverage probability is defined as
\begin{eqnarray}
\bar{P}^{(D)}_{\textrm{cov}} (\beta_k)= \b E[ {P}^{(D)}_{\textrm{cov}} ( {\bf p},\beta_k) ] =  \b E[\mathbb{P}(\textrm{SINR}_k(K,{\bf p}) \geq \beta_k)],
\end{eqnarray}
where $\beta_k$ represents the  minimum SINR value for reliable D2D link connections. Further, we define the ergodic sum rate of D2D links as
 \begin{eqnarray}
R^{(D)}=\b{E}\left[ \sum_{k=1}^K \log_2 \left( 1\!+\! \textrm{SINR}_k(K,{\bf p})  \right)\right]\!.
 \end{eqnarray}

%We summarize the notations used in this paper in Table I.
% \begin{table}[h]
%\caption{Notations of parameters }
%\centerline{
%    \begin{tabular}{c|c}
%	\hline
%	\textbf{Paramers} & Values  \\
%	\hline
%	Cell radius       & $R$\\
%         D2D link density  & $\lambda$ \\
%	Average number D2D links in the cell & $K$ \\
%	Path-loss exponent & $\alpha$ \\
%	The link distance from the $\ell$th Tx to the $k$th Rx        & $d_{k,\ell}$  \\
%	The channel gain from the $\ell$th Tx to the $k$th Rx        & $|h_{k,\ell}|^2$  \\
%	The transmit power used at the $k$th Tx & $p_k$\\
%	Target SINR of the $k$th link   &$\beta_k$\\
%	The maximum transmit power of the cellular user & $P_{\textrm{max, c}}$ \\
%	The maximum transmit power of the D2D transmitters & $P_{\textrm{max, d}}$  \\
%	$G_{\textrm{min}}$ &  $d_{k,k}^{-\alpha}=3.4988\times 10^{-7}$\\
%	The average transmit power of users & $P_{\textrm{avg}}=\exp(-1)\times P_{\textrm{max}}=36.79$ mW \\
%	Noise variance for 1MHz bandwidth & $\sigma^2$ \\
%	Transmission probability of the on-off power control & $P_s$\\
%	 \hline
%    \end{tabular}}
%\end{table}

%In this report, we will utilize the tool of stochastic geometry \cite{stoyan1995stochastic} to derive lower bounds of D2D link coverage probability $\bar{P}^{(D)}_{\textrm{cov}}$ and cellular coverage probability $\bar{P}^{(C)}_{\textrm{cov}}$% achieved by different power control techniques.
%
\section{Power Control Algorithms}
When the global CSI is available at the central controller, a centralized power control algorithm is proposed, which maximizes the SINR of the cellular link while satisfying the SINR constraints for both the cellular link and D2D links. Further, when the transmitter has CSI of the direct link of the corresponding receiver only, a distributed on-off power control algorithm is proposed.

\vspace{-0.3cm}
\subsection{Centralized Power Control}
A main difference between ad hoc networks and underlaid D2D cellular networks is that centralized power control is possible when the D2D links are managed by the BS. For other management strategies, centralized power control is able to provide an upper bound on what can be achieved with more decentralized algorithms.

Suppose that the BS has global channel state information (CSI). Under this assumption, the centralized power control problem is formulated as\begin{eqnarray}
\max_{\{p_0,p_1,\ldots,p_K \}}&& \frac{G_{0,0}p_0}{\sum_{k=1}^K G_{0,k}p_{k} +\sigma^2} \nonumber \\
 \textrm{subject to}&&  \frac{G_{0,0}p_0}{\sum_{k=1}^K G_{0,k}p_{k} +\sigma^2}  \geq \beta_0\nonumber \\
&&  \frac{G_{k,k}p_k}{G_{k,0}p_0+\sum_{\ell\neq k}^K G_{k,\ell}p_{\ell} +\sigma^2}  \geq \beta_k,  \nonumber \\
&& 0  \leq p_0 \leq { P_{\textrm{max, c}}}, \nonumber \\
&& 0  \leq p_k \leq { P_{\textrm{max, d}}}, %\nonumber\\
%
%&& 0  \leq\b E [ p_0] \leq { P_{\textrm{avg, c}}}, \nonumber \\
%&& 0  \leq \b E [ p_k]  \leq { P_{\textrm{avg, d}}},
\end{eqnarray}
where $G_{k,\ell}=|{ h}_{k,\ell}|^2d_{k,\ell}^{-\alpha}$ and $\ell\in\{1,2,\ldots,\}$ and $k\in\{1,2,\ldots,\}$. This optimization problem is compactly written in a vector form as
\begin{eqnarray}
\max_{{\bf p}}&& \frac{{\bf g}_{0}^{T}{\bf p}}{{\bf g}_{0}^{cT} {\bf p}+\sigma^2} \nonumber \\
 \textrm{subject to}&& \left({\bf I}-{\bf F}\right){\bf p} \geq {\bf b} \nonumber \\
&& {\bf 0} \leq {\bf p}\leq {\bf p}_{\textrm{max}}, \label{eq:opt}
\end{eqnarray}
where ${\bf g}_{0}^T=[G_{0,0}, 0, \ldots, 0] $, ${{\bf g}_{0}^{c}}^T=[0, G_{0,1}, G_{0,2},\ldots, G_{0,K}] $, ${\bf p}_{\textrm{max}}=[P_{\textrm{max, c}}, P_{\textrm{max, d}}, \ldots, P_{\textrm{max, d}}]^T$, and the normalized channel gain matrix ${\bf F}$ and  target SINR vector ${\bf b}$ are defined as
\begin{eqnarray}
{\bf F}_{k,\ell}&=& \left\{
\begin{array}{l l}
  0,  \quad  \quad & k=\ell,\\
  \frac{\beta_{k} G_{k,\ell}}{G_{k,k}},  \quad & k\neq \ell. \nonumber
\end{array} \right.\nonumber \\
{\bf b}&=&\left[\frac{\beta_0\sigma^2}{G_{0,0}},~\frac{\beta_1\sigma^2}{G_{1,1}},~\frac{\beta_2\sigma^2}{G_{2,2}},~\ldots,\frac{\beta_K\sigma^2}{G_{K,K}}\right]^T.\nonumber
\end{eqnarray}

Since the objective function (linear-fractional function) is quasi-convex and the constraint set  is convex (a polytope in particular) with respect to power profile vector ${\bf p}$, the optimal solution is able to be obtained by using standard convex programming tools, provided that the feasible set is nonempty. Note that the matrix ${\bf F}=[{\bf f}_0,{\bf f}_1,\ldots, {\bf f}_K]$  is comprised of nonnegative elements and is irreducible because all the active D2D links interfere each other. By the Perron-Frobenious theorem, the following well-known lemma proved in  \cite{Foschini:93} gives a necessary and sufficient condition on the feasibility of the optimization problem  (\ref{eq:opt}).
\begin{lem}\cite{Foschini:93}
The constraint set in the optimization problem (\ref{eq:opt}) is nonempty if and only if the maximum modulus eigenvalue of  ${\bf F}$ is less than one, i.e.,
$
 \rho({\bf F}) <1,
$
where $\rho(\cdot)$ denotes the spectral radius of a matrix.
\end{lem}

\begin{table*}
\caption{Proposed Centralized Power Control Algorithm }
\centerline{
     \begin{tabular}{c|c}
	\hline
	\textbf{Step} & Algorithm  \\
	\hline
	Initialization       & Set initial ${\bf F}^{\ell}$ for $\ell=0$, assuming $K$ D2D links are all active.  \\
	Step 1       & Test feasibility condition $\rho({\bf F}^{\ell})<1$. If this condition is satisfied, go to Step 5. Otherwise, go to Step 2. \\
	Step 2 & Pick the column of ${\bf F}^{\ell}$ such that $\hat{k}=\arg \max_{k\in \mathcal{K} /\{0\}}\|{\bf f}^{\ell}_k\|_2$.  \\
	%Step 3  & Make the ${\hat k}$-th D2D link inactive. \\
	Step 3   &Generate a reduced matrix ${\bf \hat{F}}^{\ell}$ by removing the $\hat{k}$-th column and row vectors in ${\bf F}^{\ell}$. \\
	Step 4 &Update ${\bf F}^{\ell+1} = {\bf \hat{F}}^{\ell}$ by increasing $\ell=\ell +1$. Go to Step 1.\\
	Step 5  & Solve the optimization problem in (\ref{eq:opt}).\\
	\hline
    \end{tabular}}
    \label{tab:alg}
\end{table*}

%Let us consider the feasible set of the cellular link outage minimizing optimization problem.  The $\ell$-th SINR constraint for D2D links forms a half-space $H_{k}^+=\{{\bf p} ~|~  {\bf f}_{k}^T{\bf p} \leq c_{k} \}$ on $\mathbb{R}^{K+1}$. Thus, the feasible set is interpreted to be an intersection of $K$ half-spaces on $\mathbb{R}^{K+1}$, i.e., a polyhedron,
%\begin{eqnarray}
%S_{\textrm{fea}}=\cap_{k=1}^{K} H_{k} =\{{\bf p} ~|~{\bf a}_{k}^T{\bf p}\leq c_k, ~k=1,2,\ldots,K\}
%\end{eqnarray}
We next describe our proposed centralized algorithm to solve the optimization problem (\ref{eq:opt}). First we assume that D2D receivers can feedback all the perfect normalized channel gains $G_{i,k}$ and target SINR information $\beta_i$  to the BS. Using this assumption, the BS then computes the transmit power used for both D2D transmitters and the uplink user. Note that the feasible set should be nonempty to obtain the optimal solution, i.e., $ \rho({\bf F}) <1$. Since the normalized channel gains $G_{i,k}$, however, are random variables (the locations of all the transmit nodes are random variables), there exists a non-zero probability that the power control solution is infeasible, i.e., $\b P ( \{\rho({\bf F}) \geq 1\})\neq 0$, especially when the number of D2D links $K$ is large. When the solution is infeasible, an admission control method is needed in conjunction with the power control algorithm to provide a feasible solution to the power control problem by selecting a subset of D2D links. This D2D link selection problem may be solved by brute-force search, which requires $\sum_{r=1}^{K}\binom{K}{r} $ computations. The computational complexity grows exponentially with $K$. Instead of brute-force search, we propose an efficient D2D link selection algorithm with low computational complexity for this problem. The key idea is to drop D2D communication links successively that causes the maximum sum of the interference power in the network until the feasibility condition is satisfied. For $K$ given D2D links, we first test feasibility condition of the optimization problem in (\ref{eq:opt}). If $ \rho({\bf F}) >1$, i.e., the feasibility set is empty, we select the $\hat{k}$-th D2D transmitter such that it creates the maximum sum of interference power to all other receivers, i.e., $\hat{k}=\arg \max_{k\in \mathcal{K} /\{0\}}\|{\bf f}_k\|_2$, and then accordingly we remove the $\hat{k}$-th row and column to reduce the size of matrix ${\bf F}$. We keep reducing the size of matrix ${\bf F}$ until the feasibility condition is satisfied. Table \ref{tab:alg} summarizes the proposed D2D link selection method in conjunction with power control.

\vspace{-0.3cm}
\subsection{Distributed On-Off Power Control Algorithm}

In this subsection, we provide a distributed power control algorithm. The distributed power control is an effective interference mitigation method that requires no coordination between transmitters; the signaling overheads for sharing CSIT is not needed. In the absence of coordination, each D2D transmitter chooses its transmit power to maximize its own rate towards its intended receiver, disregarding the interference caused to the others. The proposed on-off  method is to select the D2D transmit power from the decision set $\{0,P_{\textrm{max, d}}\}$ solely based on knowledge of the direct link information and a nonnegative threshold $G_{\textrm{min}}$ that is fixed and known by all users. Specifically, the power used by D2D pair $k$ is $P_{\textrm{max, d}}$ when the link quality is good in the sense that $|h_{k,k}|^2d_{k,k}^{-\alpha}>G_{\textrm{min}}$, and $0$ otherwise. Mathematically,
\begin{eqnarray}
p_k&=& \left\{
\begin{array}{l l}
  P_{\textrm{max, d}} & \textrm{with $P_s$} \\
 0  & \textrm{with $1-P_s$}.
\end{array} \right.
\end{eqnarray}
where $P_s$ denotes the transmit probability given by
\begin{align}
P_s&=\b{P}[|h_{k,k}|^2d_{k,k}^{-\alpha}>G_{\textrm{min}}] = \exp\left(- G_{\textrm{min}}d_{k,k}^{\alpha}\right) .
\end{align}

Note that the proposed power control method is distributed as
each D2D transmitter decides its transmit power by the own channel gain $|h_{k,k}|^2$ and threshold $G_{\rm min}$. For a given distribution of the channel gain, selecting a proper threshold $G_{\rm min}$ (the transmission probability $P_s$) plays an important role in determining the sum rate performance of the D2D links. On the one hand, choosing a large $G_{\rm min}$ (a small $P_s$), reduces the inter-D2D interference. On the other hand, larger $G_{\rm min}$ (smaller $P_s$) leads to smaller number of active D2D links within the disk. Therefore, a good choice of $G_{\rm min}$ balancing these two competing factors leads to achieve a high D2D sum rate performance. This motives us to optimize the $G_{\min}$ ($P_s$) for maximizing the D2D sum rate performance. That problem is tackled in Section V.

\textbf{Remark:} The proposed power control algorithm may be useful in non-random networks because it can be applicable in any realization of the proposed random network. Therefore, the randomness in the network modeling is not a key part of the algorithm, rather it is a component of the analysis to show that it works.

\section{Coverage Probability Analysis For Distributed Power Control}

In this section, we derive the cellular link coverage probability, propose an optimal power control strategy for the cellular link under the average transmit power constraint, and derive the D2D link coverage probability. To analyze the coverage probabilities using the tool of stochastic geometry, we assume that the transmit power of each D2D transmitter is i.i.d. with distribution function $F_{p_k} (\cdot)$ and that the transmit power of the uplink user is independent and has distribution function $F_{p_0}(\cdot)$. Note that the coverage probability analysis we provide in this section is valid for any distributed power control algorithms that select its own transmit power independently of the transmit power used at the other D2D transmitters.

\subsection{Cellular Link Coverage Probability}

%We derive analytical expression for the cellular link coverage probability $\bar{P}^{(C)}_{\textrm{cov}}$.
Assume that the BS is located at the origin. The $\sinr$ of the typical uplink is given by
\begin{align}
\sinr_0 = \frac{ p_0 |h_{0,0}|^2  d_{0,0}^{-\alpha} }{ \sum_{k \in \Phi  } p_{k} |h_{0,k}|^2  d_{0,k}^{-\alpha} + \sigma^2}.
\end{align}
Further, since the cellular user's location is distributed uniformly in the circle with radius $R$, the distribution function of the distance $d_{0,0}$ of the cellular link  is given by
\[ F_{d_{0,0}} (r) = \left\{ \begin{array}{ll}
         0 & \mbox{if $r < 0$};\\
         \frac{r^2}{R^2} & \mbox{if $0\leq r \leq R$};\\
        1 & \mbox{if $r\geq R$}.\end{array} \right. \]
The following theorem provides an analytical formula for the uplink coverage probability.
\vspace{0.2cm}
\begin{theorem}
The cellular link coverage probability is
\begin{align}
\bar{P}^{(C)}_{\textrm{cov}} (\beta_0)= \b E_{X} \left[ e^{  -a_1 X - a_2 X^{\frac{2}{\alpha}}  }  \right],
\label{eq:cellularcov}
\end{align}
where $a_1 = \sigma^2\beta_0$, $a_2 =  \frac{\pi \lambda \beta_0^{ \frac{2}{\alpha} } }{\textrm{sinc} (\frac{2}{\alpha})}   \b E \left[ p_{k}^{ \frac{2}{\alpha} }\right]$, $X = p_0^{-1} d_{0,0}^{\alpha}$ with cdf
$
F_X (x) = \int  F_{ d_{0,0} } ( x^{\frac{1}{\alpha}} p^{\frac{1}{\alpha}} ) \d F_{p_0} (p).
$
\label{Theorem1}
\vspace{0.2cm}
\end{theorem}
\begin{proof}
See Appendix \ref{Lin:pro:1:proof}.
\end{proof}
\vspace{0.2cm}

Theorem \ref{Theorem1} provides an intuition that how important network parameters affect the cellular link coverage probability. For example, we observe that $\bar{P}^{(C)}_{\textrm{cov}}(\beta_0)$ depends on two D2D-related network parameters: $\lambda$ and $\b E\left[ p_k^{\frac{2}{\alpha}}\right]$. In particular, $\bar{P}^{(C)}_{\textrm{cov}}(\beta_0)$ decreases as the density $\lambda$ of D2D transmitters increases, which is intuitive as higher D2D link density causes more interference to the cellular link. Further, the random D2D power control $p_k$ affects $\bar{P}^{(C)}_{\textrm{cov}}(\beta_0)$ only through its $\frac{2}{\alpha}$-th moment. This implies that the system can control the impact of D2D links on the cellular link by constraining $\b E\left[ p_k^{\frac{2}{\alpha}}\right]$  and then find the optimal distribution of $p_0$ to maximize cellular link coverage probability.

\vspace{2mm}
\textbf{Example 1 (A Closed Form Expression):} Let us consider the case where the uplink user uses a constant transmit power $p_0=P_{\rm max, c}$ and ignore the noise $\sigma^2=0$. With path-loss exponent  $\alpha=4$, a standard value in terrestrial outdoor wireless systems, the expression of $\bar{P}^{(C)}_{\textrm{cov}}(\beta_0)$ simplifies substantially as
\begin{align}
\bar{P}^{(C)}_{\textrm{cov}}(\beta_0)&=\int_0^R \exp\left(\frac{a_2}{\sqrt{P_{\rm max, c}}}r^2\right)\frac{2r}{R^2} \d r \nonumber\\
%&=\frac{1-\exp\left(-\frac{a_2}{\sqrt{P_{\rm max, c}}}R^2\right)}{\frac{a_2}{\sqrt{P_{\rm max, c}}}R^2} \nonumber \\
&=\frac{1-\exp\left(-\frac{\pi \lambda \sqrt{\beta_0} }{\textrm{sinc} (1/2)\sqrt{P_{\rm max, c}}}   \b E [ \sqrt{p_{k}} ]R^2\right)}{\frac{\pi \lambda \sqrt{\beta_0} }{\textrm{sinc} (1/2)\sqrt{P_{\rm max, c}}}   \b E [ \sqrt{p_{k}} ]R^2}.
\end{align}
Further, if the D2D transmitters send the signal using power $P_{\rm max, d}$ with probability $P_s=0.5$, the coverage probability becomes
\begin{align}
\bar{P}^{(C)}_{\textrm{cov}}(\beta_0)
&=\frac{1-\exp\left(-\frac{\pi (\lambda/2) R^2}{\textrm{sinc} (1/2)} \sqrt{\frac{ P_{\rm max, d}}{P_{\rm max, c}}} \sqrt{\beta_0}   \right)}{\frac{\pi (\lambda/2) R^2}{\textrm{sinc} (1/2)} \sqrt{\frac{P_{\rm max, d}}{P_{\rm max, c}}} \sqrt{\beta_0} }. \label{eq:Pcov_UE_example}
\end{align}
This expression explicitly shows that the coverage performance of the cellular link is jointly determined by three factors: 1) the average number of active D2D transmitters $\b E[K]=\pi (\lambda/2) R^2$, 2) the power ratio between the cellular and the D2D user $\sqrt{\frac{P_{\rm max, d}}{P_{\rm max, c}}} $, and 3) the target threshold $\beta_0$. To validate our analysis, we compare the coverage probability expression in (\ref{eq:Pcov_UE_example}) with the simulation result. As illustrated in Fig \ref{fig:UE_cov}, the coverage probability performance of the uplink user is well matched with the corresponding Monte Carlo simulation over the entire range of $\beta_0$ and different $\lambda \in\{0.00002,0.00005\}$.

\begin{figure}
\centering
\includegraphics[width=9cm]{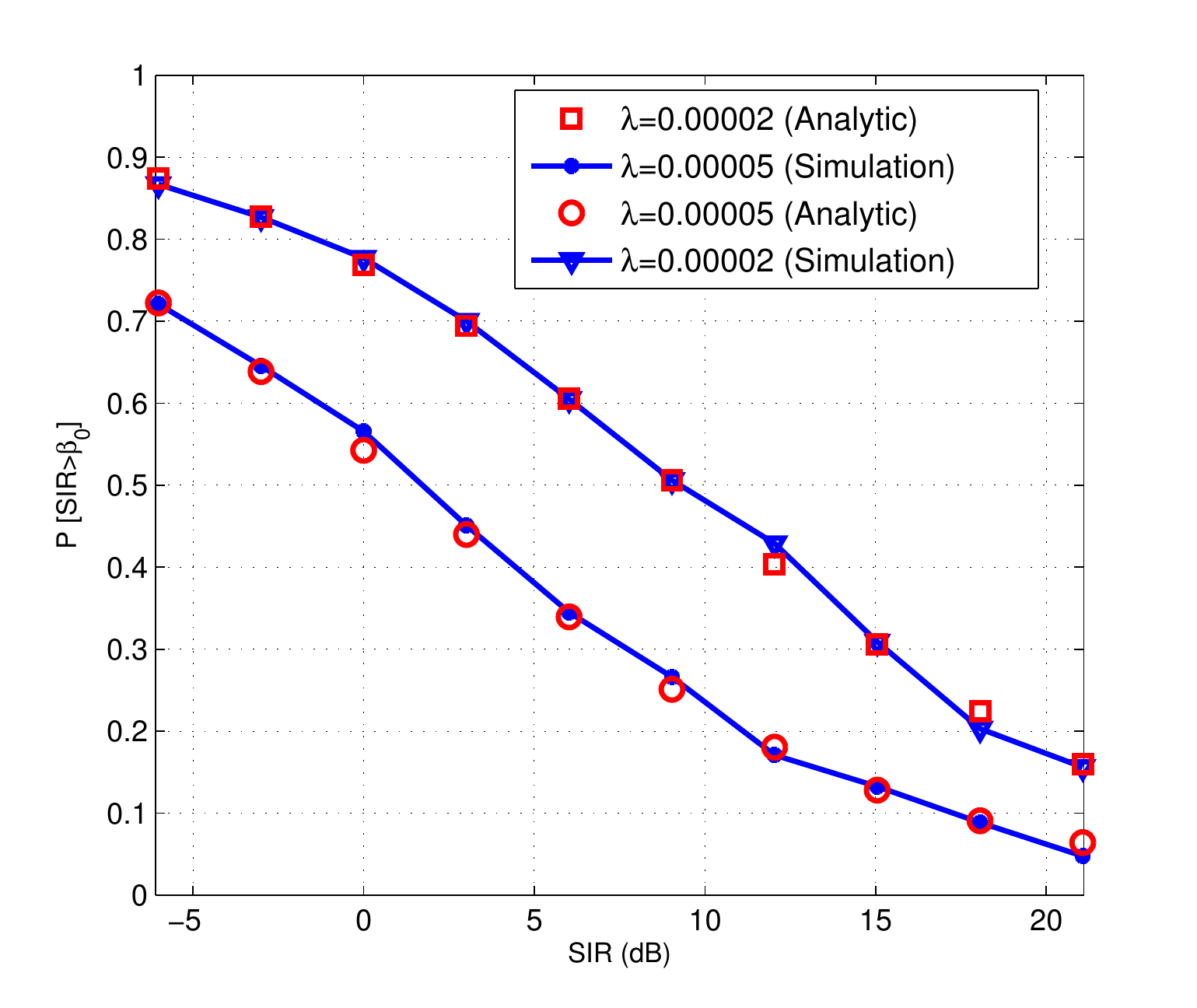}\vspace{-0.5cm}
\caption{Coverage probability performance of the uplink user with a set of parameters $d_{k,k}=50$ m, $R=500$ m, $P_{\rm max, c}=100$ mW, $P_{\rm max, d}=0.2$ mW, $\b P[G_{k,k}>G_{\rm min}=0.5]$, and $\lambda \in\{0.00002,0.00005\}$.}
\label{fig:UE_cov}\vspace{-0.4cm}
\end{figure}
%In particular,

%Second, we observe that the path loss exponent $\alpha$ affects $\bar{P}^{(C)}_{\textrm{cov}}$ in two different ways. On the one hand, $a_2$ which reflects the impact of D2D interference is monotonically decreasing for $\alpha >2$. Thus, the cellular link experiences less D2D interference with larger $\alpha$, i.e. larger $\alpha$ leads to better spatial separation between cellular and D2D links. On the other hand, $X = p_0^{-1} d_{0,0}^{\alpha}$ is proportional to the inverse of the received signal power and is monotonically increasing with respect to $\alpha$. Thus, the cellular link experiences higher signal power loss with larger $\alpha$. To summarize, larger $\alpha$ is preferred if the cellular link is interference-limited, while the converse is true if the cellular link is noise-limited.

We next provide a simple lower bound for $\bar{P}^{(C)}_{\textrm{cov}}(\beta_0)$, which is useful for an arbitrary path-loss exponent value and noise-limited case. The lower bound simply depends on the certain moments of $p_0$ and $d_{0,0}$ (rather than the distributions). This lower bound is formalized in the following corollary.
\begin{corollary}
Cellular link coverage probability $\bar{P}^{(C)}_{\textrm{cov}}$ can be lower bounded as
\begin{align}
\bar{P}^{(C)}_{\textrm{cov}} (\beta_0)&\geq \bar{P}^{(C)}_{\textrm{cov,lb}} (\beta_0) \nonumber \\
&= e^{  -a_1 \frac{2  }{2+\alpha} R^{\alpha} \cdot \b  E[ p_0^{-1}]  - a_2 (\frac{2}{2+\alpha})^{\frac{2}{\alpha}} R^2 \cdot \left( \b E[ p_0^{-1}]  \right)^{\frac{2}{\alpha}}  }.
\end{align}
\label{Lin:pro:2}
\end{corollary}
\begin{proof}
See Appendix \ref{Lin:pro:2:proof}.
\end{proof}

\subsection{Optimal Cellular Link Power Control Strategy}
In this subsection, we provide an optimal power control strategy for the cellular link when the cellular user has location information of distance $d_{0,0}$. As shown in Theorem 1, the coverage probability of the cellular link is a function of the transmit power $p_0$ and the distance $d_{0,0}$ of the uplink user. Conditioning on the location for the uplink user, i.e. $d_{0,0}=d$, the cellular link coverage probability is reduced as $\phi (p_0) = e^{ -a_1 d^{\alpha} p_0^{-1} - a_2 d^2 p_0^{- \frac{2}{\alpha}} } $. Under the average and peak power constraints of the uplink transmission power $p_0$, the optimal distribution function of $p_0$, i.e., $F_{p_{0}}$, is obtained by solving the following optimization problem:
\begin{eqnarray}
\textrm{maximize} &&  \int \phi (p_0)   \d F_{p_{0}} (p_0)  \nonumber \\
 \textrm{subject to}&& \int p_0  \d F_{p_{0}} (p_0) = P_{\textrm{avg, c}} \nonumber \\
 && \int   \d F_{p_{0}} (p_0) = 1\nonumber \\
 && p_0 \leq P_{\rm max, c}. \label{eq:optimization_cellular}
\end{eqnarray}
Note that (\ref{eq:optimization_cellular}) is an infinite-dimensional optimization problem because the distribution function $F_{p_0}$ may have an infinite number of degrees of freedom. Although this class of optimization problem is not solvable in general, we are able to find the optimal solution of the distribution function $F_{p_0}$ thanks to a remarkably simple structure for the optimization problem in (\ref{eq:optimization_cellular}), which is stated in the following theorem.
\vspace{0.2cm}
\begin{theorem}
There exists a non-degenerate $p^\star (d) \in (0,\infty)$ that maximizes $\frac{\phi(p_0)}{p_0}$. Further, conditional on $d_{0,0}=d$ the optimal power allocation strategy maximizing the cellular link coverage probability is on-off power control.
\label{theorem2}
\end{theorem}\vspace{0.2cm}
\begin{proof}
See Appendix \ref{theorem2:proof}.
\end{proof}
\vspace{0.2cm}

From Theorem \ref{theorem2}, the on-off power control strategy provides the cellular user with the optimal coverage probability performance and the optimal transmission power $p_0^\star(d)$ is a maximizer of the function $\frac{\phi(p_0)}{p_0}$. Although the exact expression of $p_0^\star(d)$ is difficult to obtain, we are able to find a closed form expression in the interference limited regime, i.e., $\sigma^2=0$.

\vspace{0.2 cm}
\begin{corollary}
For the interference limited regime ($\sigma^2=0$), the optimal transmit power of the cellular user placed at the distance $d_0=d$ with respect to the BS is
\begin{align}
p_0^\star(d) = \max\left\{ \min \left\{\tilde{p}_0,~ P_{\rm max, c}\right\},  P_{\textrm{avg, c}}\right\}. \label{eq:optimalpower}
\end{align}
\label{cor:cov}
where $\tilde{p}_0(d)=\left(\frac{2 \b E [ p_{k}^{ \frac{2}{\alpha} }] }{\alpha\textrm{sinc} (\frac{2}{\alpha})}\right)^{\!\!\frac{\alpha}{2}}\!\! \b{E}[K]^{\alpha/2} \beta_0 \left(\frac{d}{R}\right)^{\alpha}$.
\end{corollary}
\vspace{0.2 cm}

\begin{proof}
Let $x=\frac{1}{p_0}$. Then, for the interference limited regime, the objective function for the uplink power optimization problem becomes $\phi (x)x = x\exp\left( - a_2 d^2 x^{2/\alpha} \right) $.
From the first order optimality condition, i.e., $\frac{\partial \phi (x)x}{\partial x}=0$, we obtain the maximizer $x^*=(\frac{\alpha}{2a_2d^2})^{\frac{\alpha}{2}}$. Putting $a_2=\frac{\pi \lambda \beta_0^{ \frac{2}{\alpha} } }{\textrm{sinc} (\frac{2}{\alpha})} \b E [ p_{k}^{ \frac{2}{\alpha} }] = \frac{\b{E}[K]\beta_0^{ \frac{2}{\alpha} } }{R^2\textrm{sinc} (\frac{2}{\alpha})}   \b E [ p_{k}^{ \frac{2}{\alpha} }]$ and using $p_0^*=\frac{1}{x^*}$, we obtain the $\tilde{p}_0(d)$ in (\ref{eq:optimalpower}). Since the transmit power should satisfy the maximum transmit power constraint, we take the minimun value between $\tilde{p}_0(d)$ and $P_{\rm max, c}$.
Note that since the cellular user uses binary power control, the cellular user's transmit probability becomes $\frac{P_{\textrm{avg, c}} }{ p_0^{\star}(d) }\leq 1$ .
\end{proof}
\vspace{0.2 cm}

Using the optimal cellular user transmit power obtained in Corollary \ref{cor:cov}, we have a closed form expression on the cellular user coverage probability for the interference-limited regime.
\vspace{0.2 cm}
\begin{corollary}
For the interference limited regime ($\sigma^2=0$) and a given uplink user distance $d$, the cellular user coverage probability becomes as in (19).
\begin{figure*}
\begin{eqnarray}
\bar{P}^{(C)}_{\textrm{cov}}(d;\beta_0) =
 \left\{
\begin{array}{l l}
\frac{P_{\textrm{avg, c}}}{P_{\textrm{max, c}} }  \exp\!\left\{\!- \! \frac{\b{E}[K]\beta_0^{ \frac{2}{\alpha} } }{\textrm{sinc} (\frac{2}{\alpha})}   \b E [ p_{k}^{ \frac{2}{\alpha} }]\left(\frac{d}{R}\right)^2 P_{\rm max, c}^{-2/\alpha}\right\}, & \quad \textrm{for} \quad   \tilde{p}_0 (d) \geq P_{\textrm{max, c}}, \\
 \frac{P_{\textrm{avg, c}} \exp\left(-\frac{2}{\alpha}\right)}{\left(\frac{2 \b E [ p_{k}^{ \frac{2}{\alpha} }] }{\alpha\textrm{sinc} (\frac{2}{\alpha})}    \right)^{\!\!\!\alpha/2} \!\!\!\b{E}[K]^{\alpha/2} \beta_0 \left(\frac{d}{R}\right)^{\!\alpha}}, & \quad \textrm{for} \quad  P_{\textrm{avg, c}} < \tilde{p}_0 (d) < P_{\textrm{max, c}}, \\
 \exp\left\{- \frac{\b{E}[K]\beta_0^{ \frac{2}{\alpha} } }{\textrm{sinc} (\frac{2}{\alpha})}   \b E [ p_{k}^{ \frac{2}{\alpha} }]\left(\frac{d}{R}\right)^2 P_{\textrm{avg, c}}^{-2/\alpha} \right\} & \quad \textrm{for} \quad \tilde{p}_0(d) \leq P_{\textrm{avg, c}}.
\end{array} \right.\\ \hline \nonumber \label{eq:CelCov}
\end{eqnarray}
\end{figure*}
\end{corollary}
%\vspace{0.2cm}
\begin{proof}
Recall that under the assumptions of interference limited regime and the fixed distance of uplink user, the cellular user coverage probability becomes $\bar{P}^{(C)}_{\textrm{cov}}(d;\beta_0) = \b E_{p_0}\!\left[e^{ - a_2 d^2 p_0^{- \frac{2}{\alpha}} } \right]$. Since the optimal power control strategy of the uplink user is the binary power control, i.e., $p_0=p_0^\star(d) $ with probability $\frac{P_{\textrm{avg, c}} }{ p_0^{\star}(d) }$ and $p_0=0 $ with probability $1-\frac{P_{\textrm{avg, c}} }{ p_0^{\star}(d) }$, the coverage probability expression is reduced as
\begin{align}
\bar{P}^{(C)}_{\textrm{cov}}(d;\beta_0) &=  \frac{P_{\textrm{avg, c}}}{p_0^\star(d) }  \exp\left( - a_2 d^2 p_0^\star(d)^{-2/\alpha}\right).
\end{align}
Using the solution of $p_0^{\star}(d)$ in (\ref{eq:optimalpower}) and $a_2 = \frac{\b{E}[K]\beta_0^{ \frac{2}{\alpha} } }{R^2\textrm{sinc} (\frac{2}{\alpha})}   \b E [ p_{k}^{ \frac{2}{\alpha} }]$, we obtain the desired coverage probability expression.
\end{proof}
\vspace{0.2cm}

The coverage probability of the cellular user $\bar{P}^{(C)}_{\textrm{cov}}(d;\beta_0) $ behaves in three different ways according to the location of the user $d$. When $ \tilde{p}_0(d)  < P_{\textrm{avg, c}}$ (i.e., the user is located at around the cell center), the cellular user uses the constant transmit power $P_{\textrm{avg, c}}$; this results in the coverage probability decreases exponentially with respective to $\beta_0^{\frac{2}{\alpha}}$. When $ P_{\textrm{avg, c}} < \tilde{p}_0 (d) < P_{\textrm{max, c}}$ (i.e., the user is located at mid range of the cell edge), the on-off power control strategy is activated. In this regime, the cellular user increases its transmit power proportionally to $d^{\alpha}$, implying that the cellular user should increase the transmit power according to the inverse of path-loss, agreeing with intuition. Further, the uplink user is required to increase the transmit power linearly according to $\b{E}[K]^{\frac{\alpha}{2}}$ where $K$ is the random number of D2D links in the coverage of the BS. From this on-off power control strategy, the coverage probability decreases linearly with respect to the target SIR $\beta_0$.
In the regime of $\tilde{p}_0(d) \geq P_{\textrm{max, c}}$ (i.e., the user is located at around the cell edge), the cellular user sends a signal with its maximum transmit power with probability $\frac{P_{\rm avg, c}}{P_{\rm max, c}}$ due to the peak power constraint; thus, the coverage probability decreases linearly with respective to $\beta_0^{\frac{2}{\alpha}}$ again.

We provide an example to help the understanding of three different behaviors on  the coverage probability performance.

\textbf{Example 2 (Three Different Behaviors of the Cellular Link Coverage Probability)}: In this example, let us consider a set of typical parameters: the path-loss exponent $\alpha=4$, the cell radius $R=500 m$, the target SINR $\beta_0=6$ dB, the average number of D2D links $\b{E}[K]=\lambda \pi R^2=39$, and the average and maximum transmit power constraints of the cellular user $P_{\textrm{avg, c}}=0.1$ W and $P_{\textrm{max, c}}=0.2$ W. Further, we assume that the D2D links use a constant transmit power $P_{\rm max, d}=0.0001$ W, which gives us $\b{E}[\sqrt{p_k}]= \sqrt{P_{\textrm{max, d}}}=\sqrt{0.0001}$. In this set of parameters, the transmit power of the cellular link is expressed in terms of the distance $d$ as
\begin{align}
\tilde{p}_0(d)&=\left(\frac{ \b E [ p_{k}^{ \frac{1}{2} }] }{2\textrm{sinc} (\frac{1}{2})}\right)^{\!2}\!\! \b{E}[K]^{2} \beta_0 \left(\frac{d}{R}\right)^{2} \nonumber \\
&\simeq 0.375 \times \left(\frac{d}{R}\right)^{2} \textrm{W}.
\end{align}
If the cellular user is located in the half of cell radius $d=\frac{R}{2}$, the optimal uplink transmission power $p_0^\star(R/2)=\max\left\{ \min\left\{ \tilde{p}_0, P_{\textrm{max, c}} \right\}, P_{\textrm{avg, c}} \right\}=0.1$ W because of $\tilde{p}_0(R/2)\simeq 0.093 < 0.1$ W, implying that the average transmit power is used in this regime. Thus, the cellular user coverage probability becomes $\bar{P}^{(C)}_{\textrm{cov}}(R/2, \beta_0)\simeq 0.743$. Alternatively, if we consider that the cellular user is located at mid range of the cell edge with $d=0.7R$, then the transmission power of the cellular user becomes $p_0^\star(0.7R)=0.183$ W, which means that the uplink user opportunistically sends its uplink signal using transmit power 0.183 W with probability of $\frac{100}{183}$. Therefore, it gives a coverage probability performance $\bar{P}^{(C)}_{\textrm{cov}}(0.7R, \beta_0) \simeq 0.3298$. For the regime of $d=R$, the optimal transmit power of the cellular link $p_0^\star(R)=0.2$ W (the maximum transmit power); thus, the coverage probability performance in this regime is $\bar{P}^{(C)}_{\textrm{cov}}(R, \beta_0) \simeq 0.216$.

\subsection{D2D Link Coverage Probability}

We derive an expression for the coverage probability for the typical D2D link. Consider an arbitrary communication D2D pair $k$ and assume that the D2D receiver is located at the origin. Then,
\begin{align}
\sinr_k =  \frac{ p_{k} |h_{k,k}|^2  d_{k,k}^{-\alpha} }{ \sum_{x \in \Phi \setminus \{k\}  } p_{i} |h_{k,i}|^2  ||x_i||^{-\alpha} + p_0 |h_{k,0}|^2 d_{k,0}^{-\alpha} + \sigma^2},
\end{align}
where $||x_i||=d_{k,i}$. Using the same approach we used to prove Theorem 1, we need to compute two Laplace transforms $ \b E\left[ e^{-s p_0 |h_{k,0}|^2 d_{k,0}^{-\alpha} } \right] $ and  $ \b E\left[ e^{-s \sum_{x \in \Phi \setminus \{k\}  } p_{i} |h_{k,i}|^2  ||x_i||^{-\alpha}  } \right] $ to derive the distribution of $\sinr_k$.

First let us  focus on  $ \b E\left[ e^{-s p_0 |h_{k,0}|^2 d_{k,0}^{-\alpha} } \right] $.
As we assume the uplink user and the D2D receiver are randomly positioned in the disk with radius $R$, the pdf $f_{d_{k,0}} (r)$ is given by \cite{moltchanov2012distance}
\begin{align}
f_{d_{k,0}} (r) \!\!=\!\! \frac{2r}{R^2} \!\left(\!\!  \frac{2}{\pi}\! \cos^{-1}\!\! \left( \frac{r}{2R} \right) \!-\! \frac{r}{\pi R} \sqrt{ 1 \!-\! \frac{r^2}{4R^2} } \!\right), \!\!\quad 0\!\leq r \!\leq 2R. \label{eq:distance_interference}
\end{align}
Besides, $|h_{k,0}|^2$ is a random variable with the exponential distribution, i.e. $|h_{k,0}|^2~\textrm{Exp} (1)$ and $p_0$ has cdf $F_{p_0} (p)$.  Noting further that $p_0,  |h_{k,0}|^2$, and $d_{k,0}$ are independent, we have
\begin{align}
 \b E\!\left[ e^{-s p_0 |h_{k,0}|^2 d_{k,0}^{-\alpha} }\!\right] \!=\!\!\! \int \!\!\!\!\int_0^\infty \!\!\!\!\int_{0}^{2R}\!\!\!\! e^{-s p h r^{-\alpha} \!-\! h }    f_{d_{k,0}} (r) \d r    \d h \d F_{p_0} (p). \label{eq:D2D_cov}
\end{align}

With a similar approach as in the previous subsection, it is possible to derive the complementary cumulative distribution function (ccdf) of $\sinr_k$, and the coverage probability for the typical D2D link is given in the following theorem.
\vspace{0.2 cm}
\begin{theorem}\label{Th3}
The coverage probability of the typical D2D link is given by
\begin{align}
\bar{P}^{(D)}_{\textrm{cov}}(\beta)   = \b E_{Z} \left[ e^{  -b_1 Z - b_2 Z^{\frac{2}{\alpha}}  } \hat{\c L}_{Y}(\beta Z) ]  \right], \label{eq:d2d_CP}
\end{align}
where $b_1 = \sigma^2\beta$, $b_2 =  \frac{\pi \lambda \beta^{ \frac{2}{\alpha} } }{\textrm{sinc} (\frac{2}{\alpha})}   \b E [ p_k^{ \frac{2}{\alpha} }]$,  $Z = p_k^{-1} d_{k,k}^{\alpha}$ with cdf
$
F_Z (z) = \int  F_{ d_{k,k} } ( x^{\frac{1}{\alpha}} p^{\frac{1}{\alpha}} ) \d F_{p_k} (p)
$, $Y = p_0 |h_{k,0}|^2 d_{k,0}^{-\alpha}$ and $\c L_{Y}(s) =  \b E\left[ e^{-s p_0 |h_{k,0}|^2 d_{k,0}^{-\alpha} } \right]$.
\label{Theorem2}
\end{theorem}
\vspace{0.2 cm}
\begin{proof}
See Appendix \ref{Lin:pro:10:proof}.
\end{proof}
\vspace{0.2 cm}

%\begin{figure*}
%\begin{align}
%\b E\left[  \frac{1}{1+\frac{\kappa}{d_{k,0}^4}}\right]\!=\!1\!+\!\kappa \!\!\left(\!\!\frac{\sqrt{2}\sqrt{1\!+\!\sqrt{1\!+\!\frac{16R^4}{\kappa}}}\!-\!2}{4R^4} \!-\! \frac{2(1\!+\!16R^4/\kappa)^{1/4} }  { (\kappa\!+\!16R^4) \frac{ \sqrt{ \kappa\!+\!16R^4\!+\! \sqrt{ \kappa(\kappa\!+\!16R^4)} } }{\sqrt{2\kappa +32R^4}}   }\!\! \right)\!-\!\frac{40R^4}{3\kappa} \mathcal{A}(\kappa,R)
%\end{align} \label{eq:D2DCovP_closed}
%\end{figure*}

\begin{figure}
\centering
\includegraphics[width=9.5cm]{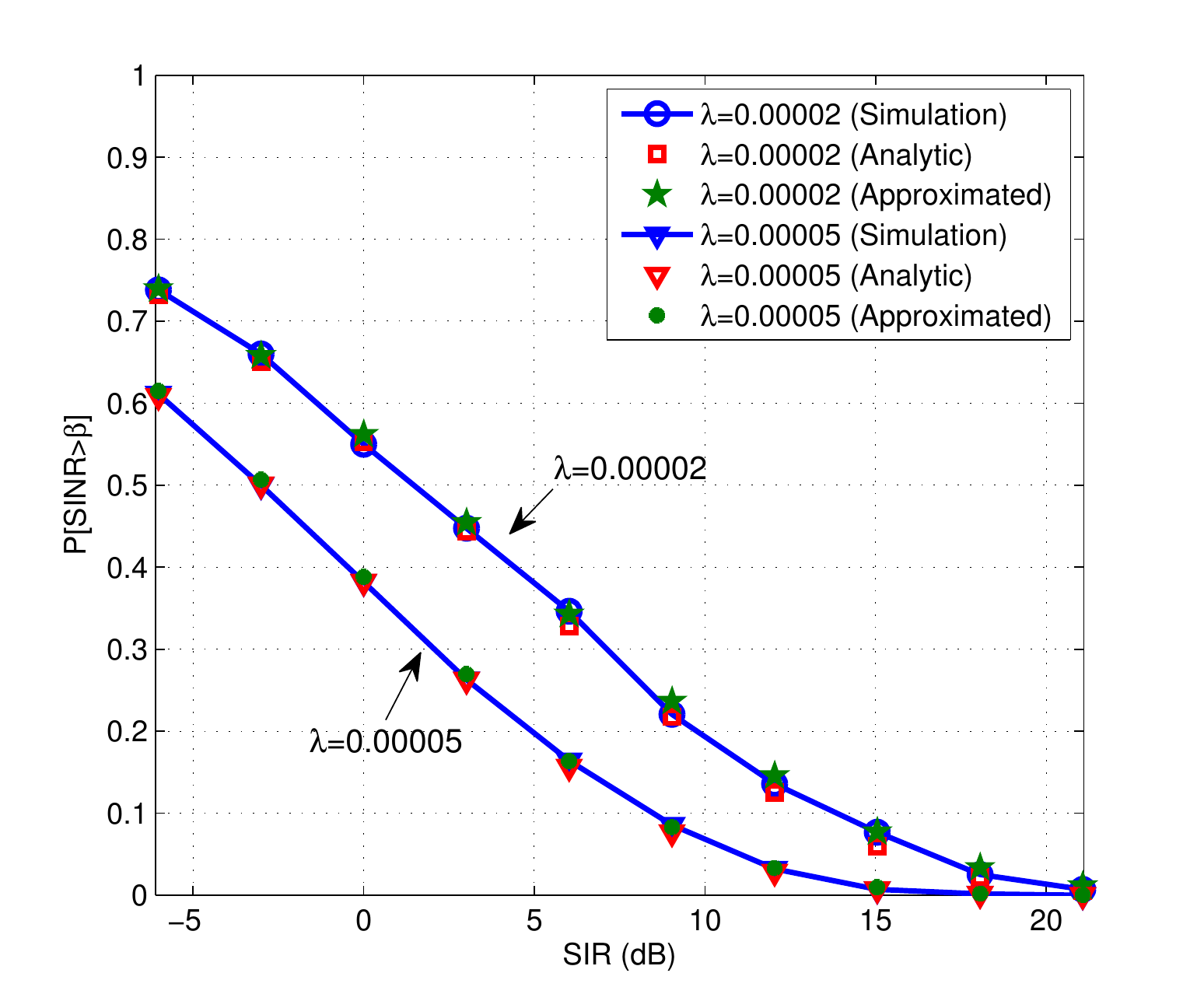}\vspace{-0.5cm}
\caption{Coverage probability performance of the typical D2D link with a set of parameters $d_{k,k}=50$ m, $p_0=100$ mW, $p_{k}=0.1$ mW, $R=500$ m, and $\lambda \in\{0.00002,0.00005\}$.}\vspace{-0.4cm}
\label{fig:D2D_cov}
\end{figure}

To shed further light on the significance of the expression derived in Theorem \ref{Th3}, it is instructive to consider a special case where all D2D transmitters communicate with their corresponding receivers with a fixed distance $d_{k,k}$ using a fixed transmit power $p_k$. When the BS uses also a constant transmit power $p_0$, we are able to derive a closed from expression of the coverage probability for the typical D2D link in the interference limited regime $\sigma^2=0$ as
\begin{align}
\bar{P}^{(D)}_{\textrm{cov}}(\beta)   &\!=\! \exp\!\left(\! -  \frac{\pi \lambda \beta^{ \frac{2}{\alpha} } }{\textrm{sinc} (\frac{2}{\alpha})} p_k^{-\frac{2}{\alpha}} p_k^{\frac{2}{\alpha}}  d_{k,k}^2\!\!\right) \!\!  \b E\!\!\left[ e^{-\beta \frac{p_0}{p_{k}} \!\left(\frac{d_{k,k}}{d_{k,0}}\right)^{\!\!\alpha} \!\! |h_{k,0}|^2  }\! \right] \nonumber \\
&= \exp\left( -  \frac{\pi \lambda \beta^{ \frac{2}{\alpha} } }{\textrm{sinc} (\frac{2}{\alpha})}   d_{k,k}^2\right)\b E\left[  \frac{1}{1+\beta \frac{p_0}{p_{k}} \left(\frac{d_{k,k}}{d_{k,0}}\right)^{\alpha} }\right], \label{eq:d2d_CP_simple}
\end{align}
%\nonumber \\
%&\leq \exp\left( -  \frac{\pi \lambda \beta^{ \frac{2}{\alpha} } }{\textrm{sinc} (\frac{2}{\alpha})}   d_{k,k}^2\right) \frac{1}{1+\beta \frac{p_0}{p_{k}} \b E\left[\left(\frac{d_{k,k}}{d_{k,0}}\right)^{\!\!\alpha}\right] }
%\end{align}
where the second equality comes from the fact that $ |h_{k,0}|^2\sim\exp(1)$ and the expectation in (\ref{eq:d2d_CP_simple}) is over $d_{k,0}$. %With a particularization of the path-loss exponent $\alpha=4$,
%\begin{align}
%\bar{P}^{(D)}_{\textrm{cov}}(\beta) &= \exp\left( -  \frac{\pi \lambda \sqrt{\beta}}{\textrm{sinc} (0.5)}   d_{k,k}^2\right)\b E\left[  \frac{1}{1+\frac{\kappa}{d_{k,0}^4} }\right],  \label{eq:Pcov_D2D_example}
%\end{align}
%where $\kappa = \beta \frac{p_0}{p_{k}}d_{k,k}^4$.

We further consider a simple but approximated expression of the coverage probability for the typical D2D link. With an approximation of $\b E\left[  \frac{1}{1+\frac{\kappa}{d_{k,0}^{\alpha}}}\right]\simeq \frac{1}{1+\frac{\kappa^{2/\alpha}}{\b E[d_{k,0}]^{2}}}$, which is obtained from numerical observations, we have an approximated expression on the coverage probability as
\begin{align}
\bar{P}^{(D)}_{\textrm{cov}} (\beta)
&\simeq \exp\left( \!\!-  \frac{\pi \lambda \beta^{ \frac{2}{\alpha} } }{\textrm{sinc} (\frac{2}{\alpha})}   d_{k,k}^2\!\!\right)\!\!\frac{1}{1\!+\!\left(\!\beta \frac{p_0}{p_{k}}\!\right)^{\!\!2/\alpha}\!\!\!\frac{d_{k,k}^{2}}{\b E[d_{k,0}]^{2}}}, \nonumber \\
&= \exp\left( -  \frac{\pi \lambda \beta^{ \frac{2}{\alpha} } }{\textrm{sinc} (\frac{2}{\alpha})}   d_{k,k}^2\right)\frac{1}{1+\left(\beta \frac{p_0}{p_{k}}\right)^{\!\!2/\alpha}\!\!\!\frac{d_{k,k}^{2}}{\b (128R/(45\pi))^{2}}},\label{eq:d2d_CP_approx}
\end{align}
where the equality follows from the first moment of $d_{k,0}$, $\b E[d_{k,0}]= \frac{128R}{45\pi}$ given in \cite{moltchanov2012distance}.

%Using the distribution of $d_{k,0}$ in (\ref{eq:distance_interference}), it is possible to compute $\b E\left[  \frac{1}{1+\frac{\kappa}{d_{k,0}^4}}\right]$ in terms of the generalized hypergeometric function as shown in (\ref{eq:D2DCovP_closed}) where $\mathcal{A}(\kappa,R)= _3\!\!\!F_2$
%$\left[\left\{\frac{3}{2},\frac{7}{4},\frac{9}{4}\right\},\left\{\frac{5}{2},\frac{5}{2}\right\}, -\frac{16R^4}{\kappa}\right]$.

To validate our analysis, we compare the analytic expressions (exact and approximated) in (\ref{eq:d2d_CP_simple}) and (\ref{eq:d2d_CP_approx}) with simulation results. Fig. \ref{fig:D2D_cov} depicts the analytical expressions alongside the result of the corresponding Monte Carlo simulation for the entire range of $\beta$ and different $\lambda$. The agreement is accurate. Further, the approximated expression on the coverage probability in (\ref{eq:d2d_CP_approx}) provides a very precise  approximated performance, especially when $\alpha=4$  case.

%This example elucidates the fact that the coverage probability of the typical D2D link is additionally degraded by the factor of $0.7081$ due to the interference of the cellular link transmission.

%As the cluster size increases, the CCDF improves because the out-of-cluster interference abates, but that would come at the expense of further overheads. This motivates us to find the optimal cluster size by incorporating the overhead for multi-cell coordination. That problem is tackled
%in Section ??.
%
%
%For example, when we consider a set of typical parameters as , the coverage probability becomes
%\begin{align}
%\bar{P}^{(D)}_{\textrm{cov}} &=\exp(-0.2467) 0.7081 =0.5533.
%\end{align}

%\begin{align}
%\b E \left[ d_{k,0}^{-\alpha}\right]=\frac{2^{1+\alpha} R^{\alpha-2} \tilde{p}\left(\frac{1+\alpha}{2}\right)}{\tilde{p}\left(2+\frac{\alpha}{2}\right)}
%\end{align}

\section{ Sum Rate Analysis of D2D Links}

In this section, we analyze the sum rate of D2D links when the proposed on-off power control is applied and characterize the optimal threshold of the on-off power control, which maximizes the sum-rate of D2D links.

\vspace{-0.3cm}
\subsection{Sum Rate of D2D Links}
Let us denote the normalized inter-D2D link interference power at the $k$-th D2D receiver as $I_k=\sum_{\ell\neq k} |h_{k,\ell}|^2 d_{k,\ell}^{-\alpha}$ for $k,\ell\in \{1,2,\ldots,|\mathcal{S}|\}$ where $|\mathcal{S}|$ denotes the number of active links selected by the proposed on-off power control algorithm, i.e., $|\mathcal{S}|=\lambda \b{P}[|h_{k,k}|^2d_{k,k}^{-\alpha}\geq G_{\textrm{min}}] \pi R^2= \tilde{\lambda}\pi R^2$. Further, let $\tilde{p}=\frac{P_{\textrm{max, c}}}{P_{\textrm{max, d}}}$ denote the transmission power ratio between the D2D transmitter and the uplink user. Assuming Gaussian signal transmission from all the active links, the distribution of the interference becomes Gaussian. Then, the achievable sum rate of D2D links is written as
\begin{align}
R^{(D)}&=  \b E\left[ \sum_{k=1}^{K}\log_2\!\left(\!\! 1\!+\! \frac{|h_{k,k}|^2d_{k,k}^{-\alpha}}{I_k + |h_{k,0}|^2d_{k,0}^{-\alpha}\tilde{p}}\! \right)\!\right] \nonumber\\
& =|\mathcal{S}| \b E\left[ \log_2\!\left(1+ \rm{SIR}_k\right) \right],\nonumber \\
& =\tilde{\lambda}\pi R^2 \times \bar{R}_{\rm d2d}.  \label{eq:d2dsumrate}
\end{align}

Using the SIR distribution of the typical D2D link given in (\ref{eq:d2d_CP_simple}), the ergodic rate of the typical D2D link can be rewritten as
\begin{align}
 \bar{R}_{\rm d2d}&= \int_{0}^{\infty} \log_2(1\!+\!x) \b P[ {\rm SIR}_k \geq x] \d x  \nonumber \\
&= \int_{0}^{\infty}\frac{P^{(D)}_{\rm cov}(x)}{1+x} \d x
\nonumber \\
&=\int_{0}^{\infty}\!\!\!\! \frac{1 }{1\!+\!x}\exp\!\left(\!\! -  \frac{\pi
\tilde{\lambda} x^{ \frac{2}{\alpha} }d_{k,k}^2 }{\textrm{sinc} (\frac{2}{\alpha})}   \!\!\right)\!\!\b E\!\!\left[ \!\! \frac{1}{1\!+\!x \tilde{p}\! \left(\!\frac{d_{k,k}}{d_{k,0}}\!\right)^{\!\!\alpha} }\!\!\right] \d x \label{eq:sum_rate}
\end{align}
where the expectation is taken over $d_{k,0}$. Using an approximation of $\b E\left[  \frac{1}{1+\frac{\kappa}{d_{k,0}^{\alpha}}}\right]\simeq \frac{1}{1+\frac{\kappa^{2/\alpha}}{\b E[d_{k,0}]^{2}}}$, we have an approximated expression of the ergodic rate of the typical D2D link in (\ref{eq:sum_rate}) as
\begin{align}
 \bar{R}_{\rm d2d}&\simeq \int_{0}^{\infty}\!\!\!\! \frac{1 }{1\!+\!x}\exp\!\left(\!\! -  \frac{\pi
\tilde{\lambda} x^{ \frac{2}{\alpha} }d_{k,k}^2 }{\textrm{sinc} (\frac{2}{\alpha})}   \!\!\right)\!\!  \frac{1}{1\!+\! \left(x \tilde{p}\right)^{\frac{2}{\alpha}} \!\!\left(\!\frac{d_{k,k}}{\b E [d_{k,0}]}\!\right)^{\!\!2} } \d x \\
&= \int_{0}^{\infty}\!\! \frac{1}{1\!+\!x}\exp\!\left(\!\! -  \frac{\pi
\tilde{\lambda} x^{ \frac{2}{\alpha} }d_{k,k}^2 }{\textrm{sinc} (\frac{2}{\alpha})}   \!\!\right)\!\!  \frac{1}{1\!+\kappa x^{\frac{2}{\alpha}}} \d x, \label{eq:D2Drates}
\end{align}
where $\kappa \!=\! \left(\frac{P_{\rm max, c}}{P_{\rm max, d}}\right)^{\frac{2}{\alpha}}\!\! \left(\!\frac{d_{k,k}}{128R/(45\pi)}\!\right)^{\!\!2}$. Interestingly, the approximated expression of the ergodic rate of the typical D2D link in (\ref{eq:D2Drates}) is determined by two factors: (1) the Laplace transform of the total interference power created by all active links on the entire network, i.e., $\exp\!\left(\!\! -  \frac{\pi
\tilde{\lambda} x^{ \frac{2}{\alpha} }d_{k,k}^2 }{\textrm{sinc} (\frac{2}{\alpha})}   \!\!\right)$ and (2) the approximated effect of the uplink interference $ \frac{1}{1\!+\!\kappa x^{\frac{2}{\alpha}} }$.

\subsection{Optimizing D2D ON-Off Threshold}

With the characterized ergodic sum rate of D2D links, we optimize the D2D on-off threshold $G_{\rm min}$ by maximizing the approximated transmission capacity of D2D links given as
\begin{align}
R^{(D)}(\beta)&=\tilde{\lambda} \pi R^2\log_2(1+\beta)\b P[{\rm SIR_k}\geq \beta]  \nonumber  \\
 &\simeq     \tilde{\lambda} \pi R^2\exp\!\left(\!\! -  \frac{\pi
\tilde{\lambda} \beta^{ \frac{2}{\alpha} }d_{k,k}^2 }{\textrm{sinc} (\frac{2}{\alpha})}   \!\!\right)\!\!  \frac{\log_2(1+\beta)}{1\!+\kappa \beta^{\frac{2}{\alpha}}}  \\
&=  {\lambda}P_s \pi R^2\exp\!\left(\!\! -  \frac{\pi
\lambda P_s \beta^{ \frac{2}{\alpha} }d_{k,k}^2 }{\textrm{sinc} (\frac{2}{\alpha})}   \!\!\right)\!\!\frac{\log_2(1+\beta)}{1\!+\kappa \beta^{\frac{2}{\alpha}}}. \label{eq:TC1}
\end{align}
To this end, we first compute the optimal transmission probability $P_s$ by solving the optimization problem:
\begin{align}
 \max~ &R^{(D)} (\beta)  \nonumber \\
 \textrm{subject to}~~  &0< P_s \leq 1 \label{eq:optimization}
\end{align}
Although the objective function is not concave, the optimal solution of $P_s$ can be obtained by using the first order optimality condition since the objective function has an unique optimum point. The first order optimality condition yields
\begin{align}
1-\frac{\pi
\lambda \beta^{ \frac{2}{\alpha} }d_{k,k}^2 }{\textrm{sinc} (\frac{2}{\alpha})}  P_s =0,
\end{align}
from which we have $P^{\star}_s=\min\left\{ \frac{\textrm{sinc} (\frac{2}{\alpha})}{\pi
\lambda \beta^{ \frac{2}{\alpha} }d_{k,k}^2 },1\right\}$. Finally, since $P_s =\b P[|h_{k,k}|^2d_{k,k}^{-\alpha}>G_{\rm min}]$, the optimal on-off threshold can be obtained as
\begin{align}
G_{\rm min}^{\star} = \frac{-\ln(P_s^{\star}) }{ d^{\alpha}_{k,k} }. \label{eq:Gmin}
\end{align}

%In general, the optimal solution of $P_s$ can be obtained using a numerical line search technique. It is possible to obtain a closed form solution of $P_s$
%for a relaxed on-off power control method in which the D2D transmitters use their peak transmit power $P_{\rm max, d}$ when they are scheduled. In this case, the interference effect of the uplink transmission becomes a constant because $\max\{P_{\rm max, d}, P_{\rm avg, d}/P_s\}=P_{\rm max, d}$. Therefore, the optimization problem in (\ref{eq:optimization}) is reduced as
%\begin{align}
% \max~ & \frac{\log_2(1\!+\!\beta)\lambda P_s \pi R^2}{1+\kappa \beta}\exp\!\left(\!\! -  \frac{\pi
%\lambda P_s \beta^{ \frac{2}{\alpha} }d_{k,k}^2 }{\textrm{sinc} (\frac{2}{\alpha})}   \!\!\right)  \nonumber \\
% \textrm{subject to}~~  &0< P_s \leq 1, \label{eq:optimization_simple}
%\end{align}
%where $\kappa \!=\! \frac{P_{\rm max, c}}{P_{\rm max, d}}\!\! \left(\!\frac{d_{k,k}}{128R/(45\pi)}\!\right)^{\!\!\alpha}$.

Using the solution of $P_s^{\star}$, the approximated transmission capacity in (\ref{eq:TC1}) can be expressed as
\begin{align}
R^{(D)}(\beta) \simeq
 \left\{
\begin{array}{l l}
 {\lambda}  \pi R^2\exp\!\left(\!\! -  \frac{\pi
\lambda \beta^{ \frac{2}{\alpha} }d_{k,k}^2 }{\textrm{sinc} (\frac{2}{\alpha})}   \!\!\right)\!\!\frac{\log_2(1+\beta)}{1\!+\kappa \beta^{\frac{2}{\alpha}}}, &  \textrm{for} ~~  \beta< \tilde{\beta}, \\
 \frac{\textrm{sinc} \left(\frac{2}{\alpha}\right)}{\exp(1)} \!\!\left(\! \frac{R}{ d_{k,k} }\! \right)^{\!\!2}\!\! \beta^{ -\frac{2}{\alpha} }\frac{\log_2(1+\beta)}{1+\kappa \beta^{\frac{2}{\alpha}}}&  \textrm{for} ~~ \beta> \tilde{\beta}.
\end{array} \right. \label{eq:transmission_capacity}
\end{align}
where $\tilde{\beta}=\left[\frac{\textrm{sinc} \left(\frac{2}{\alpha}\right) }{\pi \lambda d_{k,k}^2}\right]^{\frac{\alpha}{2}}$.
The transmission capacity of the D2D links behaves differently depending on the relative relationship between the target SINR value $\beta$ and network parameters: path-loss exponent $\alpha$ and the density of D2D links $\lambda$, and the distance of D2D link $d_{k,k}$. In the case where $\beta$ is smaller than $\left[\frac{\textrm{sinc} \left(\frac{2}{\alpha}\right) }{\pi \lambda d_{k,k}^2}\right]^{\frac{\alpha}{2}}$, all D2D transmitters are scheduled, which leads to achieve the same performance with that of no power control. Meanwhile, when $\beta$ is large enough, the D2D links are scheduled with the transmission probability $P_s^{\star}$, which results in mitigating the inter-D2D interference. In particular, in the case of  $\beta > \left[\frac{\textrm{sinc} \left(\frac{2}{\alpha}\right) }{\pi \lambda d_{k,k}^2}\right]^{\frac{\alpha}{2}}$, the transmission capacity of the D2D links becomes independent of the density of nodes $\lambda$. Further, the transmission capacity of underlaid D2D links increases linearly with the spatial packing ratio $\frac{R^2}{d_{k,k}^{2}}$ of D2D transmissions.

\begin{table}[h]
\caption{The Sum Rate Performance of D2D Links  }\centerline{
    \begin{tabular}{c|c|c|c}
	\hline
	\textbf{Density of D2D links ($\lambda$)} & 0.00001 & 0.00003 & 0.00005 \\
	\hline
	Sum Rate (Simulation) & 13.98 & 27.83  & 33.93   \\ \hline
	Sum Rate (\ref{eq:sum_analy})  & 13.63  & 26.29 & 32.54 \\
	\hline
    \end{tabular}}\label{table2}\vspace{-0.4cm}
\end{table}

By integration of the transmission capacity in (\ref{eq:transmission_capacity}) with respect to $\beta$, we have the sum rate of D2D links as
\begin{align}
R^{(D)}&\simeq \int_{0}^{\tilde{\beta}}
 \frac{{\lambda}  \pi R^2}{(1\!+\kappa x^{\frac{2}{\alpha}})(1+x)}\exp\!\left(\!\! -  \frac{\pi
\lambda x^{ \frac{2}{\alpha} }d_{k,k}^2 }{\textrm{sinc} (\frac{2}{\alpha})}   \!\!\right)\!\!  \d x \nonumber \\
&+ \int_{\tilde{\beta}}^{\infty}
 \frac{x^{ -\frac{2}{\alpha} }}{(1\!+\kappa x^{\frac{2}{\alpha}})(1+x)} \frac{\textrm{sinc} \left(\frac{2}{\alpha}\right)}{\exp(1)} \!\!\left(\! \frac{R}{ d_{k,k} }\! \right)^{\!\!2}\!\!   \d x. \label{eq:sum_analy}
\end{align}
To validate our analysis, we compare the analytical result of the D2D sum rate with that obtained through Monte Carlo simulation. Table \ref{table2} shows the sum rate performance of the D2D links under different D2D line densities when $\alpha=4$, and $d_{k,k}=50$ m, $R=500$ m, $P_{\rm max, c}=100$ mW, and $P_{\rm max, d}=0.1$ mW. It can be seen that the analytcial results well match the simulation results.

 \begin{table*}
\caption{Simulation parameters }
\centerline{
    \begin{tabular}{c|c}
	\hline
	\textbf{Paramers} & Values  \\
	\hline
	Cell radius ($R$)       & 500 (m) \\
	The D2D link range ($d_{k,k}$)       & 50 (m)  \\
         D2D link density ($\lambda$) & 0.00002 and 0.00005 \\
	Average number D2D links ($K$) & ${\b E}[K] = \pi R^2 \lambda \in\{15,39\}$ \\
	Path-loss exponent ($\alpha$)  & 4 \\
	Target SINR threshold ($\beta$)  & from -6 to 21  (dB)\\
	The maximum transmit power of the cellular user & $P_{\textrm{max, c}}=100$ mW  \\
	The maximum transmit power of the D2D transmitters & $P_{\textrm{max, d}}=0.1$ mW  \\
%	$G_{\textrm{min}}$ &  $d_{k,k}^{-\alpha}=3.4988\times 10^{-7}$\\
%	The average transmit power of users & $P_{\textrm{avg}}=\exp(-1)\times P_{\textrm{max}}=36.79$ mW \\
	Noise variance ($\sigma^2$) for 1MHz bandwidth & -143.97 (dBm) \\
	The number of realizations  & 1000 geometry drops\\	 \hline
    \end{tabular}}
\end{table*}

 \section{Simulation Results}

\begin{figure}
\centering
\includegraphics[width=9cm]{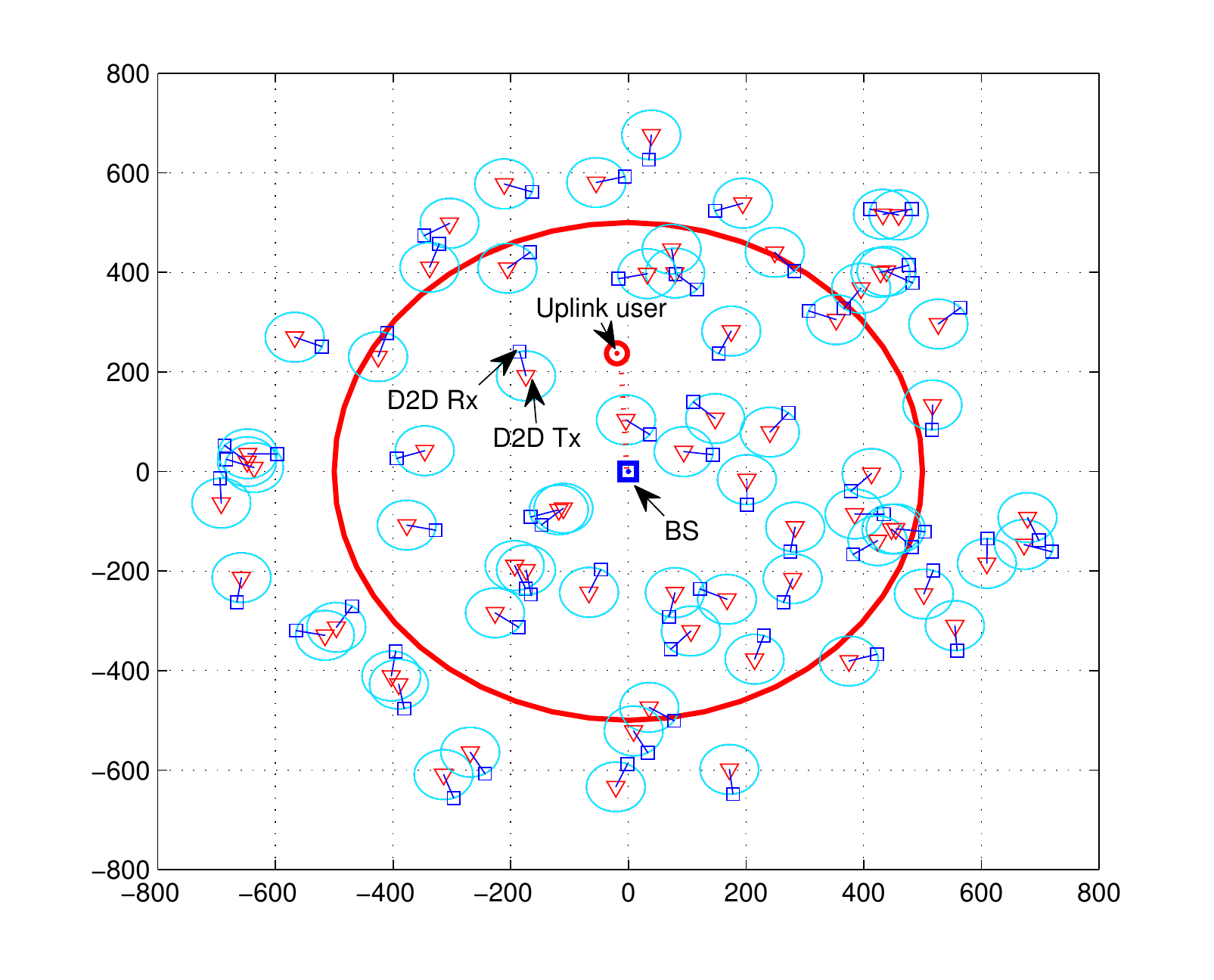}\vspace{-0.6cm}
\caption{A snap shot of link geometry for a D2D underlaid cellular network when the dense D2D link deployment scenario, i.e., $\lambda$=0.00005. In our simulation, we consider out-of-cell interference created by the D2D transmission for the warp-around effect.}\vspace{-0.4cm}
\label{fig:simulation}
\end{figure}

In this section, we provide numerical results for the D2D underlaid cellular system. From our simulation results, we first show the performance gain of the proposed power control methods compared to the no power control case in terms of the cellular user and the D2D user coverage probability.

\subsubsection{Simulation Setup}
Fig. \ref{fig:simulation}  shows one snap shot of the cell geometry. As illustrated, the BS is located at the center position $(0,0)$ in $\mathbb{ R}^2$ plane and the cellular user is uniformly dropped within the range of $R=500$ m. The D2D transmitters are dropped according to PPP with the density parameter $\lambda \in\{0.00002,0.00005\}$ in a ball centered at the origin and the radius of $R+250$ m so that the average number of D2D links are ${\b E}[K] = \pi R^2 \lambda \in\{15, 39\}$ while removing cell edge effect on the D2D link performance. Further, for a given D2D transmitter's location, the corresponding D2D receiver is isotropically dropped at a fixed distance $d_{k,k}=50$m away from the D2D transmitter. Since the D2D communication are supposed to be of short range compared to the cellular link, we assume that the average transmit power of the cellular user and D2D links are equal to $P_{\textrm{max, c}}=100$ mW and $P_{\textrm{max, d}}=0.1$ mW. Since the number of D2D links $K$ is a random variable and we evaluate the coverage probability and sum rate performance of the proposed algorithms by averaging 1000 independent realizations. Further, the optimal transmission scheduling parameter $G_{\rm min}$ is obtained as in (\ref{eq:Gmin}). The parameters used in the simulations are summarized in Table III.

\subsubsection{Coverage Probability Comparison in Sparse D2D Link Deployment}
\begin{figure}
\centering
\includegraphics[width=9cm]{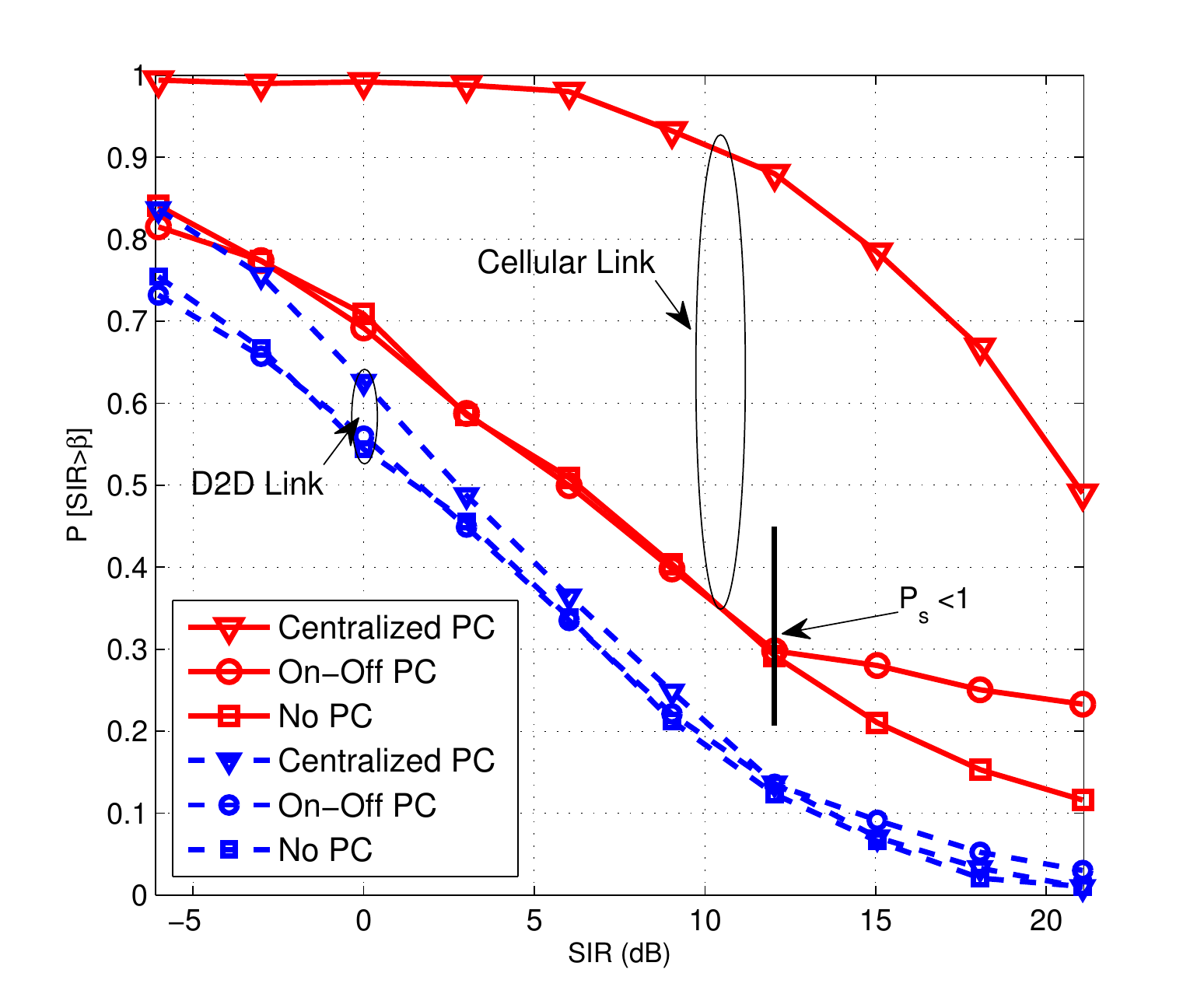}
\caption{Coverage probability performance of both the cellular and D2D  links according to different power control methods when the D2D links are sparse, i.e., $\lambda=0.00002$.}
\label{fig:sparse_cov}
\end{figure}

Suppose the sparse D2D link deployment scenario where the average number of D2D links in the cell equals ${\b E}[K] = \pi R^2 \lambda = 15.7$.
In this scenario, we compare the coverage probability of the cellular link and the D2D links under different D2D power control algorithms. As shown in Fig. \ref{fig:sparse_cov}, we observe that the proposed power control methods improve the cellular user coverage probability. The proposed power control methods also provide increased D2D link coverage probability compared to the no power control case, especially in the high target SINR regime. This implies that the power control methods are efficient to mitigate both intra-D2D and cross-tier interference when D2D links communicate with a high data rate. In particular, one remarkable observation is that the centralized power control achieves nearly perfect cellular user coverage probability performance, i.e., (no outage) in the low target SINR values, while successfully supporting a large number of active D2D links (48 \%)  when target SNIR $\beta=3$ dB. Meanwhile, the on-off distributed power control method yields performance gains for both cellular and D2D links compared to that of no power control case when the target SINR is larger than 12 dB. This is because the proposed on-off power control method provides the same performance with the no power control case until $\beta <12 $ dB, i.e., $P^{\star}_s=\min\left\{ \frac{\textrm{sinc} (\frac{2}{\alpha})}{\pi
\lambda \beta^{ \frac{2}{\alpha} }d_{k,k}^2 },1\right\}=1$ while it is activated when $\beta>12$ dB. For example, when the target SINR is 15 dB, the on-off power control method provides $13 \%$ cellular link and $5 \%$ D2D link coverage probability performance gains compared to the no power control case.

\begin{figure}
\centering
\includegraphics[width=9cm]{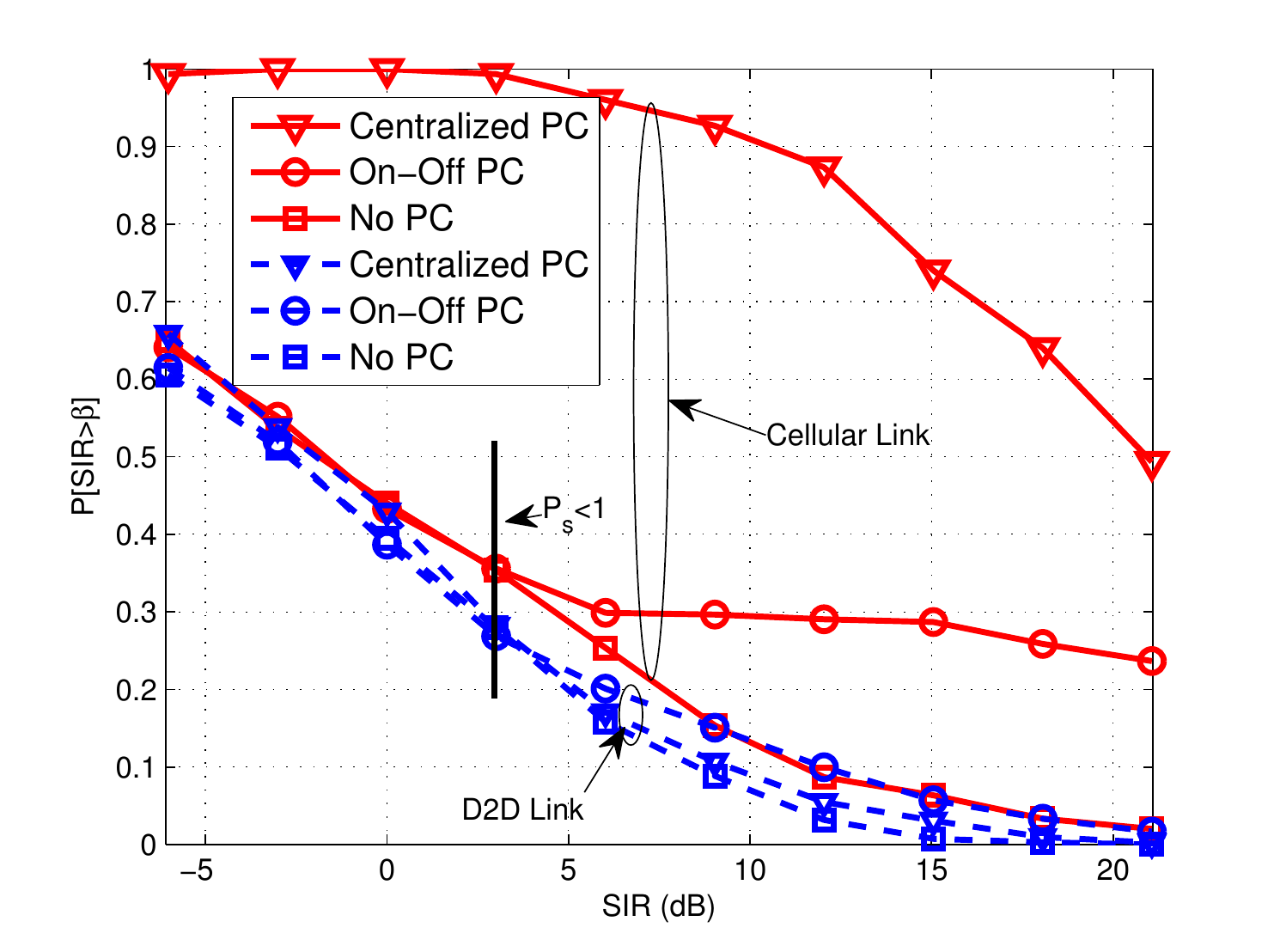}
\caption{Coverage probability performance of the cellular link according to different power control methods when the D2D links are dense, i.e., $\lambda=0.00005$.}
\label{fig:dense_cellular}
\end{figure}

\subsubsection{Coverage Probability Comparison in Dense D2D Link Deployment}
Consider a dense D2D link deployment scenario where the average number of D2D links in the cell equals ${\b E}[K] = \pi R^2 \lambda = 39$. For the dense D2D link deployment, as shown in Fig. \ref{fig:dense_cellular}, we observe  similar trends as in the sparse D2D link deployment case. One interesting point is that the performance degradation of the cellular user is not significant as the number of D2D links increases when the centralized power allocation method is applied because the proposed admission control ensures that the uplink user is protected. This implies that the centralized power control method is able to support reliable uplink performance regardless of the density of D2D links. Meanwhile, the D2D user coverage probability performance becomes deteriorated because of the increased intra-D2D link interference. It is notable that the proposed on-off power control method improves the performance of the cellular and D2D link coverage probabilities compared to that of no power control case when the target SINR is greater than $3$ dB. Although the D2D user coverage probability performance decreases in the dense scenario, the total number of successful D2D transmissions is large than that of the sparse D2D link deployment scenario. For example, when the target SINR is 3 dB, the total numbers of successful D2D transmissions in both sparse and dense scenarios are about $|\mathcal{S}|_{\textrm{sparse}} =\b E[ K {P}^{(D)}_{\textrm{cov}}(3) ] =15\times 0.5 \simeq 7.5 $ and $|\mathcal{S}|_{\textrm{dense}} =\b E[ K {P}^{(D)}_{\textrm{cov}} ( 3) ] = 39\times 0.27 \simeq 10$, respectively.

\section{Conclusions}
In this paper, we proposed a new network model  for a D2D underlaid cellular system based on stochastic geometry. In this system, we proposed both centralized and decentralized power control algorithms. One consequence of the results we observed is that the centralized power control approach leads to improve the cellular network throughput performance due to the additional underlaid D2D links while supporting reliable communication for the uplink cellular user. Meanwhile, the distributed power control approach is not enough to guarantee reliable cellular links; however, it also improves the cellular network throughput by allowing D2D links to be underlaid in the network. Future work could investigate the effect of multiple antennas at the base station, other cell interference, and joint optimization across the resource allocation and power control.

\appendix

\subsection{Proof of Theorem \ref{Theorem1}}
\label{Lin:pro:1:proof}

First notice that
\begin{align}
\bar{P}^{(C)}_{\textrm{cov}} &= \b P ( \sinr_0 \geq \beta_0  ) \nonumber \\
&= \b P \left(   \frac{ p_0 |h_{0,0}|^2  d_{0,0}^{-\alpha} }{ \sum_{k \in \Phi  } p_{k} |h_{0,k}|^2  ||x_k||^{-\alpha} + \sigma^2} \geq  \beta_0 \right)  \nonumber \\
&= \b P \left(\!  |h_{0,0}|^2\!  \geq \beta_0 p_0^{-1} d_{0,0}^{\alpha} \!\! \left( \!\sum_{k \in \Phi  } p_{k} |h_{0,k}|^2  ||x_k||^{-\alpha} \!+\! \sigma^2\! \right)\!\!  \right) \nonumber \\
&= \b E \left[ e^{ - \beta_0 p_0^{-1} d_{0,0}^{\alpha}  \left( \sum_{k \in \Phi  } p_{k} |h_{0,k}|^2  ||x_k||^{-\alpha} + \sigma^2 \right) } \right]  \nonumber \\
&\!=\!\! \b E \left[ e^{ - \sigma^2\beta_0 p_0^{-1} p_{0,0}^{\alpha} }\! \right]\!\! \b E \!\left[  e^{-\beta_0 p_0^{-1} d_{0,0}^{\alpha}\left( \sum_{k \in \Phi  } p_{k} |h_{0,k}|^2  ||x_k||^{-\alpha}   \right) } \! \right],
\label{eq:1}
\end{align}
where in the second last equality we use the fact that $|h_{0,0}|^2 \sim \textrm{Exp}(1)$ and thus $\b P(|h_{0,0}|^2\geq x) = e^{-x}$. Conditioned on the transmit power of the typical uplink transmitter $p_0 =p$ and the  distance $d_{0,0} = d$ from the cellular transmitter to  BS, we next compute the second term (\ref{eq:1}). To this end, we need the Laplace transform $\c L_{\Phi} (s) = \b E\left[ e^{-s \left( \sum_{k \in \Phi  } p_{k} |h_{0,k}|^2  ||x_k||^{-\alpha}  \right) } \right]$ given as
\begin{align}
\c L_{\Phi} (s) = e^{ - \frac{\frac{2}{\alpha} \pi^2 }{\sin (\frac{2}{\alpha} \pi ) }  \b E \left[ p_k^{ \frac{2}{\alpha} }\right] \lambda s^{ \frac{2}{\alpha} }  }.
\end{align}
Using $\c L_{\Phi} (s)$ yields
\begin{align}
\bar{P}^{(C)}_{\textrm{cov}| p_0 =p, d_{0,0}= d } =  e^{ - \sigma^2\beta_0 p^{-1} d^{\alpha} } e^{ - \frac{\frac{2}{\alpha} \pi^2 }{\sin (\frac{2}{\alpha} \pi ) } \lambda \beta_0^{ \frac{2}{\alpha} }  \b E \left[ p_k^{ \frac{2}{\alpha} }\right]  d^2 p^{- \frac{2}{\alpha} } }.
\end{align}
De-conditioning with respect to $p_0$ and $d_{0,0}$ yields the uplink coverage probability $\bar{P}^{(C)}_{\textrm{cov} }$. The last step is to derive the probability distribution of $X=p_0^{-1}d_{0,0}^{\alpha}$:
\begin{align}
F_X (x) &= \b P( p_0^{-1}d_{0,0}^{\alpha}\leq x ) \nonumber \\
&= \int \b P( d_{0,0} \leq (x p)^{\frac{1}{\alpha}} ) \d F_{p_0} (p)  \nonumber \\
&= \int  F_{ d_{0,0} } ( x^{\frac{1}{\alpha}} p^{\frac{1}{\alpha}} ) \d F_{p_0} (p).
\end{align}

\subsection{Proof of Corollary \ref{Lin:pro:2}}
\label{Lin:pro:2:proof}

Let $\phi (x) = e^{ -a_1 x - a_2 x^{\frac{2}{\alpha}} }$. We compute the first and second derivative of $\phi (x)$ as follows:
\begin{align}
\phi^{'} (x) &= - e^{ -a_1 x - a_2 x^{\frac{2}{\alpha}} } \left( a_1 + a_2 \frac{2}{\alpha} x^{\frac{2}{\alpha} - 1} \right),  \\
\phi^{''} (x) &=  e^{ -a_1 x - a_2 x^{\frac{2}{\alpha}} } \left( a_1 + a_2 \frac{2}{\alpha} x^{\frac{2}{\alpha} - 1} \right)^2 \nonumber \\ 
&+ e^{ -a_1 x - a_2 x^{\frac{2}{\alpha}} } a_2 \frac{2}{\alpha} \left(1 - \frac{2}{\alpha} \right) x^{\frac{2}{\alpha} - 2} .
\end{align}
As $\alpha >2$, $\phi^{''} (x) \geq 0$ for $x \geq 0$ and thus $\phi (x)$ is convex for $x \geq 0$. Applying Jensen's inequality, we obtain
\begin{align}
\bar{P}^{(C)}_{\textrm{cov}} = \b E_{X} \left[ e^{  -a_1 X - a_2 X^{\frac{2}{\alpha}}  }  \right] \geq e^{  -a_1 \b E[X] - a_2 \b E[X]^{\frac{2}{\alpha}}  },
\end{align}
where
$
\b E [ X ] = \b E\left[ p_0^{-1} d_{0,0}^{\alpha} \right] = \b E[ p_0^{-1}] \b E [ d_{0,0}^{\alpha} ]
$
due to the independence of $P_C$ and $D_C$. Finally, $\b E [ d_{0,0}^{\alpha} ]$ can be computed explicitly.
\begin{align}
\b E [ d_{0,0}^{\alpha} ] &= \int r^{\alpha} \d F_{d_{0,0}} (r) \nonumber \\
&= \int_{0}^{R} r^{\alpha} \frac{2r}{R^2} \d r 
\nonumber \\
&=  \frac{2}{2+\alpha} R^{\alpha}.
\end{align}

\subsection{Proof of Corollary \ref{theorem2}}
\label{theorem2:proof}

Note that $\phi (p)$ is positive-valued and continuous when $p>0$. Also, $\lim_{p\to 0^+} \frac{\phi (p)}{p} \to 0$ and $\lim_{p\to \infty}  \frac{\phi (p)}{p}\to 0$. These facts imply that there exists a non-degenerate $p^\star (d) \in (0,\infty)$ that achieves the maximum value of $\frac{\phi (p)}{p}$. If we ignore the constraint $\int   \d F_{p_{0,0}} (p) = 1$  and $p\leq P_{\rm max, c}$ for the time being and consider the following relaxed conditional cellular link coverage optimization problem:
\begin{eqnarray}
\max &&  \int \phi (p)   \d F_{p_{0,0}} (p)  \nonumber \\
 \textrm{subject to}&& \int p  \d F_{P_C} (p) = P_{\textrm{avg, c}}.
\end{eqnarray}
Let $\d L(p) = \frac{p}{P_{\textrm{avg, c}}}  \d F_{p_{0,0}} (p)$. Then the above optimization problem is equivalently formulated as
\begin{eqnarray}
\max &&  P_{\textrm{avg, c}}\cdot \int \frac{ \phi(p) }{p}  \d L(p)  \nonumber \\
 \textrm{subject to}&& \int  \d L(p) = 1. \label{eq:cov_lowbound}
\end{eqnarray}
whose optimal solution is $ L^\star (p^\star (d)) - L^\star (p^{\star-} (d)) = 1 $ and $L^\star (p) = 0 $ for $p \neq p^\star (d)$. Therefore, we conclude that the binary power control strategy is optimal.

\subsection{Proof of Theorem \ref{Theorem2}}
\label{Lin:pro:10:proof}

To prove Theorem 3, we need to derive the ccdf of $\sinr_k$. To this end, using Slivnyak's theorem \cite{stoyan1995stochastic}, it is easy to see that
\begin{align}
\c L_{\Phi \setminus \{k\}} (s) &= \b E\left[ e^{-s { \sum_{x \in \Phi \setminus \{k\}  } p_{i} |h_{k,i}|^2  ||x_i||^{-\alpha}   } } | k \in \Phi \right]
= \c L_{\Phi} (s) \nonumber\\ &
= e^{ - \frac{ \pi \lambda }{\textrm{sinc} (\frac{2}{\alpha}  ) }  \b E \left[ p_k^{ \frac{2}{\alpha} }\right]  s^{ \frac{2}{\alpha} }  }.
\end{align}
It follows that
\begin{align}
&\b P ( \sinr_k \geq \beta  ) \\\nonumber 
&= \b P \left(    \frac{ p_{k} |h_{k,k}|^2  d_{k,k}^{-\alpha} }{ \sum_{x \in \Phi \setminus \{k\}  } p_{i} |h_{k,i}|^2  ||x_i||^{-\alpha} \!+\! p_0 |h_{k,0}|^2 d_{k,0}^{-\alpha} + \sigma^2} \geq  \beta \right)   \\
\!\!&\!=\!\! \b P \!\!\left(\!\!  |h_{k,k}|^2  \!\geq \!\beta p_k^{-1} d_{k,k}^{\alpha} \left( \!\!\sum_{x \in \Phi \setminus \{k\}  }\!\!\!\!\! p_{i} |h_{k,i}|^2  ||x_i||^{-\alpha} \!\!+\! p_0 |h_{k,0}|^2 d_{k,0}^{-\alpha} \!+\! \sigma^2  \!\right)\!\!\!  \right)  \\
&= \b E \left[ e^{ - \beta p_k^{-1} d_{k,k}^{\alpha}  \left( \sum_{x \in \Phi \setminus \{k\}  } p_{i} |h_{k,i}|^2  ||x_i||^{-\alpha} + p_0 |h_{k,0}|^2 d_{k,0}^{-\alpha} + \sigma^2  \right)  } \right]  \\
%&= \b E \left[ e^{ - \sigma^2  \beta p_k^{-1} d_{k,k}^{\alpha} } \right] \b E \left[  e^{-\beta p_k^{-1} d_{k,k}^{\alpha} \left(  \sum_{x \in \Phi \setminus \{k\}  } p_{i} |h_{k,i}|^2  ||x_i||^{-\alpha}  \right) }  \right] \b E[ e^{- \beta p_k^{-1} d_{k,k}^{\alpha} \cdot p_0 |h_{k,0}|^2 d_{k,0}^{-\alpha}  } ] \\
&=  \b E_{p_k^{-1} d_{k,k}^{\alpha}} \left[  e^{ - \sigma^2 \beta p_k^{-1} d_{k,k}^{\alpha}  } \c L_{\Phi} ( \beta p_k^{-1} d_{k,k}^{\alpha} ) \c L_{Y}( \beta p_k^{-1} d_{k,k}^{\alpha} ) \right],
\end{align}
which completes the proof.

%\begin{thebibliography}{1}

\bibliographystyle{IEEEtran}
\bibliography{IEEEabrv,D2D_ref_namyoo}

%\end{thebibliography}

\end{document}